\documentclass[journal,draftcls,onecolumn,12pt,twoside]{IEEEtran}
\usepackage{cite}
\usepackage{lipsum}
\usepackage{graphicx}
\usepackage{amssymb}
\usepackage{amsmath, nccmath}
\usepackage{amsthm}
\usepackage{mathtools}
\usepackage{cuted}
\usepackage{multicol}
\usepackage{color}
\usepackage{soul}

\let\labelindent\relax  
\usepackage{enumitem}

\newcommand*{\J}{\jmath}%

\DeclareSymbolFont{myletters}{OML}{ztmcm}{m}{it}
\DeclareMathSymbol{\uplambda}{\mathord}{myletters}{"15}

\newtheorem{proposition}{Proposition}

\newtheorem{theorem}{Theorem}
\normalsize
\allowdisplaybreaks

\ifCLASSINFOpdf
\else
\fi

\usepackage[tight,footnotesize]{subfigure}

\begin{document}

\title{Impact of Subcarrier Allocation and User Mobility on the Uplink Performance of Multi-User Massive MIMO-OFDM Systems}

\author{Abhinav~Anand and 
        Chandra R.~Murthy,~\IEEEmembership{Senior Member,~IEEE}
\thanks{The authors are with the Department
of Electrical Communication Engineering, Indian Institute of Science, Bangalore, 560012, India.  E-mails: \{abhinavanand, cmurthy\}@iisc.ac.in.}
}

\maketitle

\begin{abstract}

This paper considers the uplink performance of a multi-user massive multiple-input multiple-output orthogonal frequency-division multiplexing (MIMO-OFDM) system with mobile users. Mobility brings two major problems to a MIMO-OFDM system: \emph{inter carrier interference} (ICI) and \emph{channel aging}. In practice, it is common to allot multiple contiguous subcarriers to a user as well as schedule multiple users on each subcarrier. Motivated by this, we consider a general subcarrier allocation scheme and derive expressions for the ICI power, uplink signal to interference plus noise ratio and the achievable uplink sum-rate, taking into account the ICI and the multi-user interference due to channel aging. We show that the system incurs a near-constant ICI power that depends linearly on the ratio of the number of users per subcarrier to the number of subcarriers per user, nearly independently of how the UEs distribute their power across the subcarriers. Further, we exploit the coherence bandwidth of the channel to reduce the length of the pilot sequences required for uplink channel estimation. We consider both zero-forcing and maximal-ratio combining at the receiver and compare the respective sum-rate performances. In either case, the subcarrier allocation scheme considered in this paper leads to significantly higher sum-rates compared to previous work, owing to the near-constant ICI property as well as the reduced pilot overhead.

\end{abstract}


\IEEEpeerreviewmaketitle

\section{Introduction}\label{sec:introduction}

\IEEEPARstart{M}{assive} multiple-input multiple-output (MIMO) has evolved as a key technology for 5G and beyond, offering a substantial increase in the spectral and energy efficiency of cellular systems \cite{bjrnson_magazine_2016, marzetta_twc_2010, hoydis_jsac_2013}. Having a large number of antennas at the access point (AP) (i.e., base station (BS)) effectively combats fading, as the effective channel gain becomes nearly constant due to the phenomenon of channel hardening~\cite{marzetta_book_2016}. This enables each AP to serve tens to hundreds of users using the same time-frequency resource via spatial multiplexing. On the other side, the combination of MIMO and orthogonal frequency-division multiplexing (OFDM) has shown to provide high data rates and increased system flexibility \cite{xiao_tvt_2015, ng_tcom_2012, xu_tcom_2013}, and has been deployed in standards such as 3GPP LTE Advanced, IEEE 802.16 WiMax and 5G New Radio.

Much of the recent literature elucidating the performance advantages offered by single-carrier massive MIMO systems assumes the availability of perfect channel state information (CSI) at the AP or considers the impact of channel estimation errors and pilot contamination when the channels remain static over time. This inherently implies that users in the cell remain fixed at their locations or move at sufficiently low velocities. However, as user mobility increases, there arises a difference between the estimated channel at the AP and the actual channel experienced by the data symbols, a phenomenon popularly known as \emph{channel aging} \cite{akp_tvt_2017, akp_twc_2015, chopra_twc_2018, truong_jcn_2013}. This mismatch grows with time and results in a significant loss in the achievable sum-rate~ \cite{truong_jcn_2013, chopra_letters_2016}. In \cite{kong_tcom_2015}, the authors analyse the impact of channel aging and prediction on the uplink of a single-carrier massive MIMO system with MRC and ZF receivers. The design and comparison of various channel predictors for time-varying massive MIMO channels is provided in \cite{kim_tcom_2021}. The performance loss can be overcome to some extent using data-aided channel tracking in single-carrier systems~\cite{Chopra_TSP_2021}. The effect of channel aging in cell-free massive MIMO systems was analyzed in~\cite{Chopra_ComLet_2021}. 

With OFDM, mobility brings an additional impairment to the cellular system. The frequency offset resulting from the Doppler shift disrupts the orthogonality between subcarriers, resulting in \textit{inter carrier interference} (ICI) between them (see \cite{zhang_tcom_2017} and the references therein). Not only does ICI contribute an additive interference term to the received data signal, it also causes additional channel estimation error, further impacting data detection performance. Although OFDM and massive MIMO have been the dominant technology for wireless access in the past decade, surprisingly, the system performance where the two are simultaneously employed has not been studied much in the existing literature, especially in the context of channel aging and ICI. A key paper in this area is the previous work by Zhang et al. \cite{zhang_tcom_2017}, where the degradation in the sum-rate performance of a MIMO-OFDMA system was analyzed. However, that work focused on the simple case where each user is assigned only one subcarrier. In practical systems, each user is typically scheduled on multiple subcarriers in order to improve the per-user throughput, and multiple users are scheduled on each subcarrier. The latter is particularly important in massive MIMO systems, as multiple antennas at the AP can be used to suppress interference and provide array gains, and allow one to exploit the multiplexing gain offered by massive MIMO systems. The analysis in \cite{zhang_tcom_2017} does not extend to these scenarios. In this paper, we analyze the effect of channel aging and ICI in the more general scenario alluded to above. In the process, we also develop a new, low-overhead channel estimation scheme and provide novel insights into the system performance.

\subsection{Our Contributions}
In this paper, we analyze the uplink performance of a multi-user massive MIMO-OFDM (MU-mMIMO-OFDM) cellular system when mobile users transmit data on multiple contiguous subcarriers. The main contributions of the paper are as follows:
\begin{itemize}[leftmargin=*]
    \item We introduce a static subcarrier allocation scheme that generalizes the allocation adopted in the previous work \cite{zhang_tcom_2017}, allowing a UE to transmit on multiple subcarriers while also allowing a subcarrier to serve multiple UEs. We note that a static allocation of subcarriers is reasonable in massive MIMO systems due to  the channel hardening effect~\cite{marzetta_book_2016}. It is also appropriate in OFDM based systems due to the near-constant ICI property, which is an important observation in this work. These two factors imply that the effective SINR is uniform across different subcarriers, and therefore the number of subcarriers allotted to a given user is more important than which specific subcarriers are allotted.
    \item We derive an expression for the ICI power and examine its properties. For example, we show that when the number of subcarriers is large, the ICI power is nearly independent of the (transmit) power allocation employed by the users across the subcarriers allotted to them.
    \item Inspired by techniques in standards such as IEEE 802.11a and LTE, we present a pilot sequence transmission scheme that exploits the frequency-domain channel coherence to reduce the amount of training overhead involved in channel estimation. As we will see, this dramatically improves the  achievable rate of OFDM systems, especially under fast varying channels.
    \item We consider zero-forcing and maximal-ratio combining at the AP and derive expressions for the sum-rate performance. In particular, we exploit the near-constant ICI power alluded to above to derive simple yet accurate expressions for the achievable uplink SINR and the sum-rate. In turn, this allows us to obtain key insights on the sum-rate performance such as the number of subcarriers to be allotted to users, the effect of channel aging and the pilot overhead.
    \item We provide extensive numerical simulation results to validate the analytical expressions and provide further insights into the system performance.
\end{itemize}

Our results elucidate the impact of user mobility in massive MIMO-OFDM systems. In particular, we show that the number of subcarriers allotted per user needs to be judiciously chosen based on the coherence bandwidth, coherence time, number of subcarriers available, number of users to be served, and the number of antennas. By doing so, the multi-user interference and ICI can be better controlled, leading to significant gains in the achievable sum-rate. Furthermore, despite the ICI introduced by the user mobility, the benefits of massive MIMO can still be extracted, especially in low to medium mobility scenarios.

The organization of the paper is as follows: Section \ref{sec:systemmodel} presents the system model. Section \ref{sec:channelestimation} discusses channel estimation,  presents a pilot sequence allocation scheme, and analyzes the minimum pilot length required to ensure no pilot contamination. In section~\ref{sec:ici}, an expression for the ICI power is derived and analyzed. Section \ref{sec:sumrateperformance} presents the sum-rate performance of the system under two different receive combining schemes, namely, zero-forcing and maximal-ratio combining. Section \ref{sec:numericalresults} presents numerical results, and Section~\ref{sec:conclusion} concludes the paper. 

\textbf{Notation}: Boldface uppercase and lowercase letters represent matrices and vectors, respectively. The $(i, j)$-th element of the matrix $\mathbf{A}$ is denoted by $a_{ij}$. The notations $\left(\cdot\right)^{H}$ and $\left(\cdot\right)^{\dagger}$ represent the conjugate transpose and the pseudo-inverse operation, respectively, and $\mathbb{E}\left[\cdot\right]$ denotes the expectation operator. The notation $\mathcal{CN}\left(0, \sigma^{2}\right)$ denotes the circularly symmetric complex Gaussian distribution with mean zero and variance $\sigma^{2}$. To avoid any conflict arising with use of indices, we use $\J$ to represent $\sqrt{-1}$.

\section{System Model}\label{sec:systemmodel}
We consider the uplink of a single cell MU-mMIMO-OFDM system. An access point (AP) with $N_{\mathcal{B}}$ antennas (indexed as $m \in \{1, \dots, N_{\mathcal{B}}\}$) is located at the centre of the cell. There are a total of $N_{\mathcal{R}}$ single antenna user equipments (UEs) moving in random directions within the cell. The system deploys $N_{\mathcal{G}}$ subcarriers (indexed as $i \in \{1, \dots, N_{\mathcal{G}}\}$) spaced $\Delta f$ Hz apart and spanning a total bandwidth $B$ Hz. Each subcarrier is assigned to a group consisting of $N_\mathcal{U}$ UEs, where $N_\mathcal{U} \leq N_\mathcal{R}$, since one cannot schedule more than the total number of UEs on a given subcarrier. Contrariwise, each UE transmits its data on a subset of $N_{\mathcal{C}}$ contiguous subcarriers. Thus, a UE may be served by multiple subcarriers and a subcarrier may be shared among multiple UEs. However, each group of subcarriers serves a distinct group of UEs. This is ensured by creating equal number of groups on both sides, i.e.,  
\begin{equation}\label{eq:system-model}
    \frac{N_\mathcal{R}}{N_\mathcal{U}} = \frac{N_\mathcal{G}}{N_\mathcal{C}} = L, \quad L \in \mathbb{Z}^{+}  
\end{equation}
where $L$ denotes the number of UE and subcarrier groups (see figure \ref{fig:allocation-map}). As a consequence, for a given total number of UEs ($N_{\mathcal{R}}$) and total number of subcarriers ($N_{\mathcal{G}}$), the quantities $N_{\mathcal{U}}$ and $N_{\mathcal{C}}$ can be varied keeping their ratio $\frac{N_\mathcal{U}}{N_\mathcal{C}}$ constant. We assume that the system operates in the massive MIMO regime, so that $N_\mathcal{U} \ll N_\mathcal{B}$, although the total number of UEs in the cell $N_\mathcal{R}$ could be comparable to or even larger than $N_\mathcal{B}$. Table \ref{tab:notations} lists the key notations used in this paper.

\begin{table}[t]
\centering
\caption{\label{tab:notations}Notation}
\begin{tabular}{| l | c |} \hline
Total no. of AP antennas & $N_{\mathcal{B}}$ \\ \hline
Total no. of UEs & $N_{\mathcal{R}}$ \\ \hline 
Total no. of subcarriers & $N_{\mathcal{G}}$ \\ \hline 
No. of UEs allotted to each subcarrier & $N_{\mathcal{U}}$ \\ \hline 
No. of subcarriers allotted to each UE & $N_{\mathcal{C}}$ \\ \hline
No. of subcarriers in one coherence bandwidth & $N_{\mathcal{H}}$ \\ \hline 
Length of the pilot sequence & $N_{\mathcal{P}}$ \\ \hline 
No. of subcarriers on which each UE transmits pilots & $N_{\mathcal{V}}$ \\ \hline
No. of data symbols per frame & $N_{\mathcal{D}}$ \\ \hline
\end{tabular}
\end{table}

\begin{figure}
    \centering
    \includegraphics[scale=0.5]{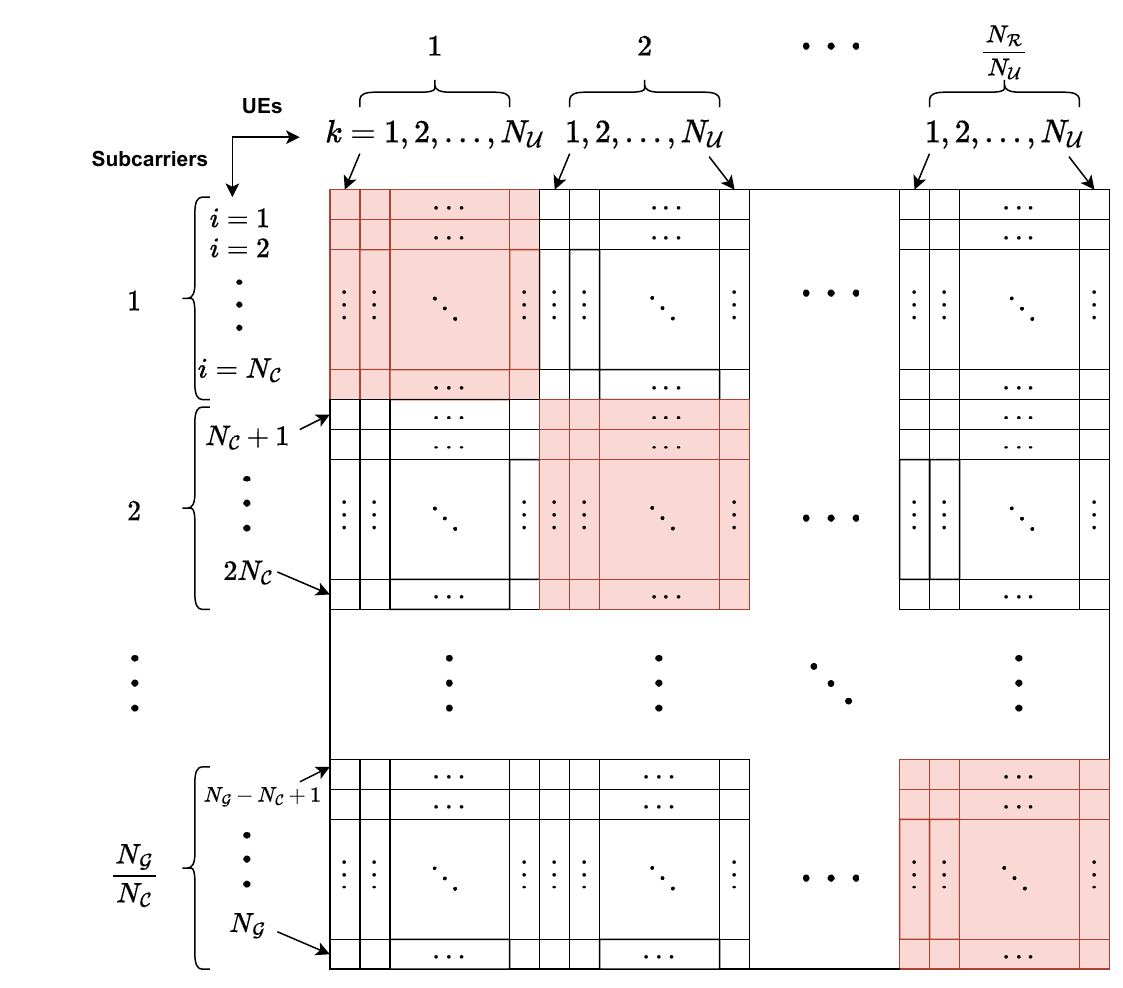}
    \caption{The UE-subcarrier allocation considered in this work.}
    \label{fig:allocation-map}
\end{figure}

The $k$-th UE served by the $i$-th subcarrier is denoted UE$_{ik}$ where $k \in \{1, \dots, N_\mathcal{U}\}$. The physical parameters concerning UE$_{ik}$ include $r_{ik}$ which represents its distance from the AP, $v_{ik}$ which denotes its instantaneous speed and $\phi_{ik}$ which models the angle between the line-of-sight (LoS) vector and the user's velocity vector. As in previous work, we assume that $\{v_{ik}\}$ and $\{\phi_{ik}\}$ are independent and identically distributed random variables drawn from a uniform distribution.\footnote{Note that $r_{ik}$, $v_{ik}$ and $\phi_{ik}$ could be the same for multiple values of $i$, if a UE is assigned multiple subcarriers. Thus, $v_{ik}$ and $\phi_{ik}$ are i.i.d. across different $(i,k)$ pairs only when they represent the speed and angle of \emph{different} UEs.} Specifically, $v_{ik}$ is uniformly distributed over $\left[0, V_{\max}\right]$ and $\phi_{ik}$ is uniformly distributed over $[0, 2 \pi]$, where $V_{\textrm{max}}$ denotes the maximum speed of all the UEs in the system\cite{zhang_tcom_2017}. 

Each UE divides its total transmit power among the subcarriers allotted to it during data transmission. If $P_{\textrm{T}, ik}$ denotes the \emph{total} transmit power of UE$_{ik}$, the RF signal transmitted by UE$_{ik}$ during the $n$-th signaling interval on subcarrier $f_{i}$ is expressed as\footnote{We alert the reader that $P_{\textrm{T}, ik}$ is \emph{not} the transmit power of a UE on subcarrier $i$. Also, as before,  $P_{\textrm{T}, ik}$ could be the same for multiple values of $i$ if a UE is assigned multiple subcarriers.}

\begin{equation}
    x_{ik}(t) = \sqrt{\eta_{ik} P_{\textrm{T},ik}} \, x_{ik}[n] \, e^{\J 2 \pi f_{i} t}.
\end{equation}
Here, $\eta_{ik}$ denotes the power control coefficient of the UE's $i$-th subcarrier, constrained as $0 \leq \eta_{ik} \leq 1$ and $\sum\limits_{i=lN_{\mathcal{C}} + 1}^{(l+1)N_{\mathcal{C}}} \eta_{ik} = 1$ for all $l \in \{0, 1, \dots, L-1\}$ and all $k \in \{1, 2, \dots, N_{\mathcal{U}}\}$. The quantity $x_{ik}[n]$ denotes the data symbol transmitted by UE$_{ik}$ during the $n$-th transmission; it is assumed to satisfy $\mathbb{E}[|x_{ik}[n]|^2] = 1$. Also, for ease of analysis, we assume that both the transmitter and the receiver employ rectangular pulse-shaping at their ends.

The average signal power received at each antenna of the AP from UE$_{ik}$ is  $cr^{-\beta}_{ik}P_{\textrm{T},ik}$, where $cr^{-\beta}_{ik}$ represents the large-scale path loss with $c$ being the path loss at a reference distance and $\beta$ being the path loss exponent \cite{zhang_tcom_2017}. To simplify the analysis, we consider path-loss-inversion-based uplink power control at the UEs \cite{dai_icassp_2006, chopra_twc_2018} so that $cr_{ik}^{-\beta} P_{\textrm{T},ik} = P_{\textrm{T}} $ for  all $i \in \{1, \dots, N_{\mathcal{G}}\}$ and $k \in \{1, \dots, N_{\mathcal{U}}\}$. The quantity $P_{\textrm{T}}$ is referred to as the effective transmit power of a UE. 

At the AP, the received signal is processed using standard OFDM operations involving the removal of cyclic prefix and the application of the fast Fourier transform (FFT). The Doppler shift induced by user mobility leads to frequency offsets between the received signal frequency and the local oscillator frequency at the AP. These frequency offsets are different for different users, and leads to ICI at the AP. Following the footsteps of \cite{zhang_tcom_2017}, with some algebra, the signal received by the $m$-th AP antenna on subcarrier $f_{i}$ for the $n$-th transmission can be expressed in frequency domain as 
\begin{equation}\label{eq:rim}
    r_{im}[n] = \sum_{k=1}^{N_{\mathcal{U}}} \sqrt{\eta_{ik} P_{\textrm{T}}}\,x_{ik}[n]\,h_{imk}[n] + u_{im}[n] + n_{im}[n].
\end{equation}
Here, $h_{imk}$ denotes the Rayleigh fast-fading channel coefficient between UE$_{ik}$ and the $m$-th AP antenna that can be modeled as a zero mean complex Gaussian random variable with variance $\sigma_{\textrm{h}}^{2}$. The term $u_{im}$ accounts for the total ICI received on subcarrier $f_{i}$ from all other subcarriers owing to the loss of orthogonality due to the Doppler-induced frequency offset. Note that $u_{im}$ equals zero when all UEs in the system are stationary. Finally, $n_{im} \sim \mathcal{CN}\left(0, \sigma_{\textrm{n}}^{2}\right)$ represents the complex AWGN at the $m$-th AP antenna.

Collectively, the signals received across all $N_{\mathcal{B}}$ antennas of the AP can be expressed in matrix-vector form as
\begin{equation}\label{eq:ri}
    \mathbf{r}_{i}[n] = \sqrt{P_{\textrm{T}}} \,\mathbf{H}_{i}[n] \, \mathbf{D}_{\eta_{i}}^{1/2} \, \mathbf{x}_{i}[n] + \mathbf{u}_{i}[n] + \mathbf{n}_{i}[n]
\end{equation}
where $\mathbf{H}_{i} \in  \mathbb{C}^{N_{\mathcal{B}} \times N_{\mathcal{U}}}$ represents the uplink channel between UEs and the AP on the $i$-th subcarrier. The vectors $\mathbf{r}_{i} = \left[{r_{i1}[n], \dots, r_{iN_{\mathcal {B}}}[n]}\right]^{T}$ and $\mathbf{x}_{i}[n] = \left[{x_{i1}[n], \dots, x_{iN_{\mathcal{U}}}[n]}\right]^{T}$ contain the frequency domain symbols received across all $N_{\mathcal{B}}$ antennas and transmitted by all $N_{\mathcal{U}}$ UEs on the $i$-th subcarrier, respectively. The diagonal matrix $\mathbf{D}_{\eta_{i}} \triangleq \textrm{diag}\{\eta_{i1}, \dots, \eta_{iN_{\mathcal{U}}}\}$ contains the power control coefficients of all the UEs within the $i$-th subcarrier. The terms $\mathbf{u}_{i}[n] = \left[u_{i1}[n], \dots, u_{iN_{\mathcal{B}}}[n]\right ]^{T}$ and $\mathbf{n}_{i}[n] = \left[n_{i1}[n], \dots, n_{iN_{\mathcal{B}}}[n]\right]^{T}$ represent the ICI and the AWGN at the AP, respectively.

\subsection{Time Varying Channel Model}
As mentioned earlier, in addition to the Doppler shift, user mobility results in a channel that varies continuously with time. These temporal variations, in turn, result in a mismatch between the channel that the AP estimates during uplink training and the channel through which the subsequent data symbols propagate. To model this disparity, let $h_{imk}^{\mathcal{P}}$ denote the fading channel coefficient for the pilots of the $i$-th subcarrier between the $k$-th UE and the $m$-th AP antenna. To simplify analysis, we assume that the channel stays constant during the pilot transmission phase \cite{akp_tvt_2017,akp_twc_2015,kong_tcom_2015, zhang_tcom_2017}. The channel coefficient experienced by the subsequent data symbol is denoted by $h_{imk}^{\mathcal{D}}[n]$, where $n = 1, \dots, N_{\mathcal{D}}$ denotes the symbol transmission index. The superscripts $\mathcal{P}$ and $\mathcal{D}$ are used to convey that the channel experienced by the pilot symbols is used as a reference for the channel experienced by the data symbols. Both the channel coefficients are zero mean random processes that are assumed to be related as follows~\cite{chopra_twc_2018}:
\begin{equation}\label{eq:aging_model}
    h_{imk}^{\mathcal{D}}\left[n\right] = \rho_{ik}\left [n\right]h_{imk}^{\mathcal{P}} + g_{imk}^{\mathcal {D}}\left[n\right].
\end{equation}
In the above expression, $\rho_{ik}\left[n\right] = J_{0}\left({2\pi f_{\textrm {D}}^{\left({ik}\right)}nT_{\textrm{s}}}\right)$ represents the Jakes' discrete-time normalized autocorrelation coefficient \cite{jakes_wiley_1974} between $h_{imk}^{\mathcal{P}}$ and $h_{imk}^{\mathcal{D}}\left[n\right]$, with $J_{0}(.)$ representing the zeroth order Bessel function of the first kind, $T_{\textrm{s}}$ representing the OFDM symbol duration,  $f_{\textrm {D}}^{\left({ik}\right)}=\frac {v_{ik}}{\textrm{c}}f_{\textrm{c}}$ denoting the maximum Doppler spread of UE$_{ik}$ with $f_{\textrm{c}}$ as the carrier center frequency and $\textrm{c}$ the speed of light. The term $g_{imk}^{\mathcal{D}}\left[n\right]$ in \eqref{eq:aging_model} represents the channel variation due to aging and is uncorrelated with $h_{imk}^{\mathcal{P}}$; it is a zero mean Gaussian random variable with variance $\sigma^{2}_{\textrm{h}}\left({1-\rho_{ik}^{2}\left [n\right]}\right)$ \cite{vu_jsac_2007}. Note that, in \eqref{eq:aging_model}, the statistics of the innovation component are chosen to ensure that the channel correlation coefficients match with those of the Jakes' model. Such a model has been used in~ \cite{chopra_twc_2018,akp_tvt_2017,zhang_tcom_2017,vu_jsac_2007, zheng_twc_2021}.

Now, we can rewrite \eqref{eq:aging_model} in matrix form as

\begin{equation}\label{eq:ar2}
    \mathbf{H}_{i}^{\mathcal{D}}\left[{n}\right] = \mathbf {H}_{i}^{\mathcal{P}} \boldsymbol{\Lambda}_{i}\left [n\right] + \mathbf{G}_{i}^{\mathcal{D}}\left [n\right],
\end{equation}
where $\mathbf{H}_{i}^{\mathcal{D}}\left[{n}\right], \mathbf {H}_{i}^{\mathcal{P}} \in \mathbb{C}^{N_{\mathcal{B}} \times N_{\mathcal{U}}}$ represent the  channel corresponding to the data and  pilot symbols, respectively. The diagonal matrix $\boldsymbol{\Lambda}_{i}\left[n\right] = \textrm{diag}\{\rho_{i1}\left[n\right], \ldots , \rho_{iN_{\mathcal{U}}}\left[n\right]\}$ contains the Jakes' correlation coefficients, and $\mathbf {G}_{i}^{\mathcal{D}}\left[n\right] \in \mathbb{C}^{N_{\mathcal{B}} \times N_{\mathcal{U}}}$ captures the channel variation due to aging. We note that the correlation coefficients in $\boldsymbol{\Lambda}_{i}\left[n\right]$ are functions of the UE velocities which are random, so the correlation coefficients and the matrix $\boldsymbol{\Lambda}_{i}\left[n\right]$ are also random.

\section{Pilot Length Reduction and Channel Estimation}\label{sec:channelestimation}
\subsection{Exploiting the Coherence Bandwidth to Reduce the Pilot Length}

In order to detect the data, the AP needs to estimate the channel from all the $N_{\mathcal{U}}$ UEs assigned to each subcarrier. In general, if the channel could vary arbitrarily across subcarriers, this entails the use of orthogonal pilot sequences by the UEs on each subcarrier allotted to them, in order to avoid pilot contamination. This, in turn, implies that $N_{\mathcal{P}} \geq N_{\mathcal{U}}$, where $N_{\mathcal{P}}$ denotes the length of the pilot sequence. Therefore, increasing $N_{\mathcal{U}}$ is necessarily accompanied by an increase in the training overhead. In this work, we exploit the coherence bandwidth, i.e., the fact that channels corresponding to  subcarriers within a coherence bandwidth are approximately equal, to reduce the training overhead needed for channel estimation without pilot contamination. Although the idea of transmitting pilots on a subset of subcarriers has been used in standards such as IEEE 802.11a and LTE, data transmission in these standards are as per OFDMA, i.e., only one user transmits pilots on a given subcarrier. In our work, we account for a MU-MIMO scenario. Specifically, we analyze the minimum pilot length required to estimate channels without incurring pilot contamination, when the number of subcarriers alloted to a UE, the number of UEs served by a subcarrier, and the number of subcarriers in a coherence bandwidth are given.

The coherence bandwidth ($B_{\textrm{c}}$) is defined as the frequency interval over which the channel seen by multiple contiguous subcarriers is approximately the same \cite{marzetta_book_2016, bjrnson_book_2017}.\footnote{The coherence bandwidth depends on the delay spread of the channel, which typically varies slowly over time. For the purposes of this work, we  consider the maximum delay spread across all users, and use it to define the coherence bandwidth. We also note that there are other definitions for the coherence bandwidth, e.g., the bandwidth over which the channel correlation coefficient remains above a threshold, say $0.7$. Since we assume that the channel can be well approximated as remaining \emph{constant} within the coherence bandwidth, we consider a more conservative threshold for the correlation coefficient, e.g.,~$0.95$.} Thus, if multiple subcarriers fit inside a coherence bandwidth, there is no need to estimate the channel on every subcarrier allotted to a UE. Instead, each UE only needs to estimate the channels on subcarriers allotted to it that lie in different coherence bandwidth intervals. This observation can be used to reduce the number of UEs transmitting pilots per subcarrier, thereby reducing the minimum pilot length required, as described in the following paragraph.

For example, if the coherence bandwidth spans $N_{\mathcal{H}} = 4$ subcarriers, and each user is allotted $N_{\mathcal{C}} = 4$ contiguous subcarriers, and each subcarrier serves $N_{\mathcal{U}} \le 4$ UEs, then it is sufficient for each user to transmit a single pilot symbol on one of the $4$ subcarriers allotted to it, with each UE transmitting its pilot on a distinct subcarrier. Next, if each subcarrier is allotted to $>4$ but $\le 8$ UEs, then two users will transmit two pilot symbols (in consecutive OFDM symbols) on one of the $4$ subcarriers allotted to it, and a distinct subset of UEs transmit pilots on each subcarrier. Further, the pilot sequence (of length $2$) allotted to each user is orthogonal to the pilot sequence allotted to the other user transmitting pilots on the same subcarrier. On the other hand, if each user is allotted, say, $8$ contiguous subcarriers and each subcarrier is allotted to $\le 4$ UEs, then it is sufficient for each UE to transmit a single pilot symbol on two of the $8$ subcarriers allotted to it (on two subcarriers that lie in different coherence bandwidth intervals), while the other UEs who are allotted the same set of $8$ subcarriers transmit their pilots on other subcarriers within the same coherence bandwidth to avoid pilot contamination. We generalize these examples in the following proposition.

\begin{proposition} 
Let $N_{\mathcal{H}}$ denote the number of contiguous subcarriers in one coherence bandwidth. Also suppose each UE is allotted $N_{\mathcal{C}}$ subcarriers while each subcarrier serves $N_{\mathcal{U}}$ UEs. Then, the  length of the pilot sequence required for estimating the channels at the AP without pilot contamination is at most 
\begin{equation}\label{eq:npfinal}
    N_{\mathcal{P}} = \left \lceil{\frac{N_{\mathcal{U}}}{\textrm{min}\left(N_{\mathcal{C}}, N_{\mathcal{H}}\right)}} \right \rceil,
\end{equation}
and it is sufficient for each UE to transmit pilots on ${N_{\mathcal{V}} \triangleq \lceil{N_{\mathcal{C}}/N_{\mathcal{H}}}\rceil}$ distinct subcarriers.
\end{proposition}

\begin{proof}
In the setting described in the Proposition, two possibilities arise: 

\textbf{Case 1}: $N_{\mathcal{H}} = k N_{\mathcal{C}}, \  k \geq 1$. In this case, instead of requiring all $N_{\mathcal{U}}$ UEs to transmit pilots on every subcarrier, it is sufficient if $\left \lceil {\frac{N_{\mathcal{U}}}{N_{\mathcal{C}}}} \right \rceil$ users transmit pilots on each subcarrier. Thus, one can set the length of the pilot sequence as
\begin{equation}\label{eq:np1}
    N_{\mathcal{P}} = \left \lceil{\frac{N_{\mathcal{U}}}{N_{\mathcal{C}}}} \right \rceil.
\end{equation}
This value is slightly suboptimal in the sense if $N_{\mathcal{U}}$ were equal to 1, we could have set $N_{\mathcal{P}} = 1$ for one of the $N_{\mathcal{C}}$ subcarriers assigned to a UE and $N_{\mathcal{P}}$ = 0 for other subcarriers. Each UE transmits pilots on only one of the $N_{\mathcal{C}}$ subcarriers alloted to it.

\textbf{Case 2}:  $N_{\mathcal{C}} = k N_{\mathcal{H}}, \  k \geq 1$. In this case, it is sufficient if each UE transmits pilots on at least one subcarrier in each coherence block. But one coherence block can be shared among $N_{\mathcal{U}}$ users. Thus, one can set

\begin{equation}\label{eq:np2_pr}
    N_{\mathcal{P}} = \left \lceil{\frac{N_{\mathcal{U}}}{N_{\mathcal{H}}}} \right \rceil.
\end{equation}
Combining the two cases above, the minimum pilot sequence length required to estimate the channels without pilot contamination is at most equal to
\begin{equation}\label{eq:npfinal_proof}
    N_{\mathcal{P}} = \left \lceil{\frac{N_{\mathcal{U}}}{\textrm{min}\left(N_{\mathcal{C}}, N_{\mathcal{H}}\right)}} \right \rceil.
\end{equation}
Note that, in either case, each UE transmits pilots on ${N_{\mathcal{V}} = \lceil{N_{\mathcal{C}}/N_{\mathcal{H}}}\rceil}$ subcarriers. 
\end{proof}

\subsection{Channel Estimation}

Now, we focus on channel estimation. Recall that $N_{\mathcal{U}}$ UEs are allotted to each subcarrier, but not all $N_{\mathcal{U}}$ UEs transmit pilots on every subcarrier. Instead, at most $N_{\mathcal{P}}$ UEs transmit pilots on a given subcarrier. The UEs transmitting their pilots on subcarrier $i$ use an orthogonal pilot book $\boldsymbol {\Phi }_{i}\in \mathbb{C}^{N_{\mathcal{P}} \times N_{\mathcal{P}}}$  to estimate the channel, where $\boldsymbol {\Phi }_{i}$ satisfies $\boldsymbol {\Phi }_{i}\boldsymbol {\Phi }_{i}^{H}=N_{\mathcal {P}}\mathbf {I}_{N_{\mathcal {P}}}$. Each column of $\boldsymbol{\Phi }_{i}$ contains the $N_{\mathcal{P}}$-length pilot sequence transmitted by a different user, and each user transmits its pilot in the $i$-th subcarrier over $N_{\mathcal{P}}$ consecutive OFDM symbols. We assume that each participating UE divides its transmit power equally among the $N_{\mathcal{V}}$ subcarriers on which it is scheduled to transmit pilots.

To estimate the channel on subcarrier $i$, the AP correlates the received pilot signal with each of the $N_{\mathcal{P}}$ pilot sequences and stacks the channel estimates of the $N_{\mathcal{U}}$ UEs that share subcarrier $i$, to obtain the least-squares (LS) estimate
\begin{equation}
    \hat {\mathbf {H}}_{i}=  \mathbf {H}_{i}^{\mathcal {P}} + \frac{1}{N_{\mathcal{P}}} \sqrt{\frac{N_{\mathcal{V}}}{P_{\textrm{T}}}} \left ({\mathbf {U}_{i}'+\mathbf {N}}_{i}'\right ) = \mathbf {H}_{i}^{\mathcal {P}} + \mathbf {G}_{i}^{\mathcal {P}} \in \mathbb{C}^{N_{B} \times N_{\mathcal{U}}}, \label{eq:estimate2}
\end{equation}
where the columns of $\mathbf{U}_{i}' \in \mathbb{C}^{N_{B} \times N_{\mathcal{U}}}$ and $\mathbf{N}_{i}' \in \mathbb{C}^{N_{B} \times N_{\mathcal{U}}}$ are independent and have the same distribution  as that of $\mathbf{u}_{i} \in \mathbb{C}^{N_{B}}$ and $\mathbf{n}_{i} \in \mathbb{C}^{N_{B}}$ in \eqref{eq:ri}, respectively, due to the orthogonal nature of the pilots. The term $\mathbf {G}_{i}^{\mathcal {P}}$ in the above expression captures the error incurred during channel estimation due to both ICI and AWGN (we discuss more about $\mathbf {G}_{i}^{\mathcal {P}}$ later in the sequel). Note that the channel estimation error is uncorrelated with the channel estimate.

Figure \ref{fig:nmse} shows a simulated plot of the normalized mean square error (NMSE) in channel estimation drawn as a function of the pilot SNR for multi-carrier and single-carrier systems at $V_{\textrm{max}} = 25~\textrm{m/s}$. The parameters chosen to generate this figure are: $N_{\mathcal{B}} = 256$, $N_{\mathcal{U}} = 4$, $N_{\mathcal{C}} = 1$, $\Delta f = 10$~kHz, and $f_{\textrm{c}} = 3$~GHz. The NMSE is computed by taking the ratio between the Frobenius norm of the channel estimation error and the Frobenius norm of the true channel (the ICI component is absent in the single-carrier system). The NMSE is independent of the number of AP antennas, since each antenna can estimate the channel independently in this case.
In a single-carrier system, the AP can obtain accurate channel estimates provided the pilot symbols are recieved at sufficiently high power. In contrast, in a multi-carrier system, user mobility induces ICI. Due to this, even at high pilot power, the channel estimates at the AP remain noisy because the power in the ICI also scales with the transmit power. As a result, we see an error floor in the NMSE with increasing pilot power.

\begin{figure}
    \centering
    \includegraphics[width=4.5in]{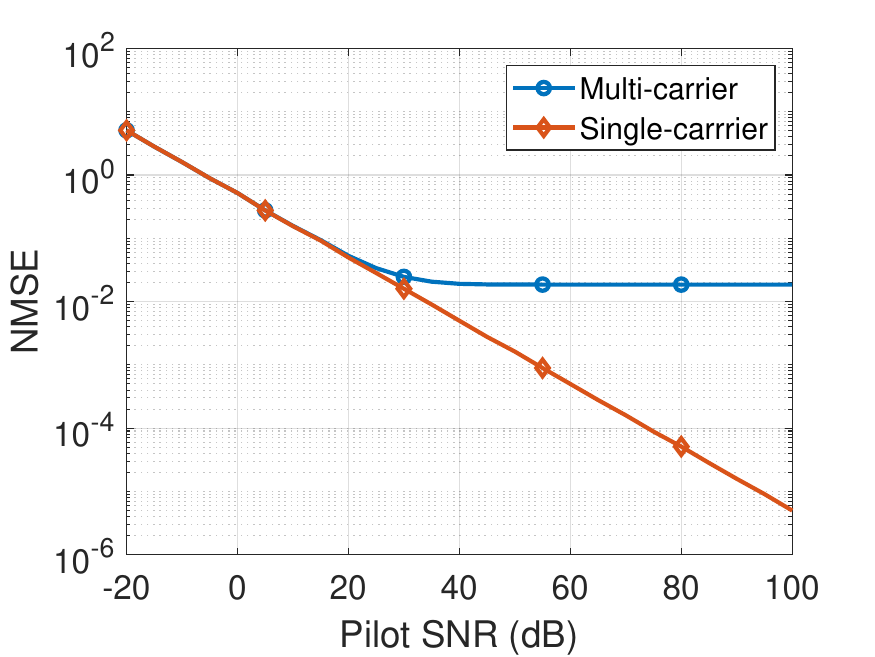}
    \caption{Plot of NMSE vs. pilot SNR for single-carrier and multi-carrier systems at $V_{\textrm{max}} = 25$m/s with $N_{\mathcal{B}} = 256$, $N_{\mathcal{U}} = 4$, $\Delta f = 10~\textrm{kHz}$ and $f_{\textrm{c}} = 3~\textrm{GHz}$.}
    \label{fig:nmse}
\end{figure}

\section{ICI Analysis}\label{sec:ici}
Now, let us focus our attention to ICI. Due to the frequency offset, the fraction of the total signal power that leaks from subcarrier $i$ onto subcarrier $j$ can be computed as~\cite{zhang_tcom_2017}
\begin{equation}
    L_{i}(f_{j})\triangleq \frac{\int _{0}^{2\pi }\int_{0}^{V_{\textrm {max}}}\textrm {sinc}^{2}\left({\left({f_{i}\,-\,f_{j}+\frac{v}{\textrm {c}}f_{\textrm{c}}\cos(\psi)}\right)T_{\textrm{s}}}\right)dvd\psi}{2\pi V_{\textrm {max}}},
\end{equation}
where $\textrm{sinc}(x) \triangleq \frac{\textrm{sin}(\pi x)}{\pi x}$. Now, the total ICI power received by the $i$-th subcarrier from UEs transmitting on subcarrier $f_{j}$ is $P_{\textrm{ICI},ijm} = P_{\textrm{T}} \sum_{k=1}^{N_{\mathcal{U}}}\eta_{jk}L_{j}(f_{i})$. Contrariwise, the ICI power recieved by subcarrier $f_{j}$ from subcarrier $f_{i}$ is given by $P_{\textrm {ICI},jim} = P_{\textrm{T}} \sum_{k=1}^{N_{\mathcal{U}}} \eta_{ik} L_{i}(f_{j})$. While it can be shown that $L_{i}(f_{j}) = L_{j}(f_{i})$, the quantities $P_{\textrm {ICI}, jim}$ and $P_{\textrm {ICI},ijm}$ may not always be equal since the equality $\sum_{k=1}^{N_{\mathcal{U}}} \eta_{jk} = \sum_{k=1}^{N_{\mathcal{U}}} \eta_{ik}$ may not always hold true. Thus, contrary to the result presented in \cite{zhang_tcom_2017} for the case where each user transmits on a single subcarrier, when users are assigned multiple subcarriers, two subcarriers may not always incur the same amount of ICI power from each other. However, the quantity that is of interest to us is the sum total of ICI contributions from all subcarriers onto a target subcarrier, say the $i$-th subcarrier, at the $m$-th AP antenna. This is computed as follows:
\begin{align}
    P_{\textrm {ICI},im} =& \, \sum _{\substack {j=1\\ j \neq i}}^{N_{\mathcal {G}}}P_{\textrm {ICI},ijm} \\
    =& P_{\textrm{T}} \sum _{\substack {j=1\\ j \neq i}}^{N_{\mathcal {G}}} \sum_{k=1}^{N_\mathcal{U}} \eta_{jk} L_{i}( f_{j} ) \\ 
    =& \, P_{\textrm{T}} \sum_{k=1}^{N_\mathcal{U}} \left( \sum_{j=1}^{N_\mathcal{G}} \eta_{jk} L_{i}( f_{j} ) - \eta_{ik} L_{i}(f_{i})\right). \label{eq:pici1}
\end{align}
Figure \ref{fig:ici_distr} shows subplots of $P_{\text{ICI},im}$ drawn as a function of the target subcarrier index $i$ for different values of  $V_{\textrm{max}}$ and $N_{\mathcal{U}}$, with the UEs dividing their power across their allotted subcarriers uniformly and at random such that $0 \leq \eta_{jk} \leq 1$ and $\sum\limits_{j=lN_{\mathcal{C}} + 1}^{(l+1)N_{\mathcal{C}}} \eta_{jk} = 1$ for all $l \in \{0, 1, \dots, L-1\}$ and all $k \in \{1, 2, \dots, N_{\mathcal{U}}\}$. The parameter values used to generate the plot are as follows: $N_{\mathcal{R}} = 2048\,\textrm{users}$, $N_{\mathcal{G}} = 512\,\textrm{subcarriers}$, $P_{\textrm{T}} = 10\,\textrm{dB}$, $T_{\textrm{s}} = 10^{-4}\,\textrm{sec}$ and $f_{\textrm{c}} = 3\,\textrm{GHz}$. When the total number of subcarriers $N_{\mathcal{G}}$ is sufficiently large, we observe that $P_{\textrm{ICI}, im}$ is roughly constant regardless of how individual UEs allocate their transmit power among the subcarriers. Moreover, the variations in the ICI curves that result from allocating transmit power randomly to the subcarriers reduce as the value of $N_{\mathcal{U}}$ increases. This lends the ICI power a deterministic nature, i.e., it does not matter how UEs distribute their powers to their subcarriers, the ICI power will remain fixed as long as the other parameters like $V_{\textrm{max}}$, $f_{\textrm{c}}$ and $T_{\textrm{s}}$ remain fixed.

To obtain an accurate closed form expression for the ICI power, we use the above observation, assume uniform power allocation across all subcarriers and substitute $\eta_{jk} = \eta_{ik} = \frac{1}{N_{\mathcal{C}}}$ in \eqref{eq:pici1}. We also remark that equal power allocation across subcarriers is near-optimal in the massive MIMO regime, as the effective channel is roughly constant across the different subcarriers due to channel hardening. This gives
\begin{figure}
    \centering
    \subfigure[$N_{\mathcal{U}}=8$]
    {
        \includegraphics[width=3.0in]{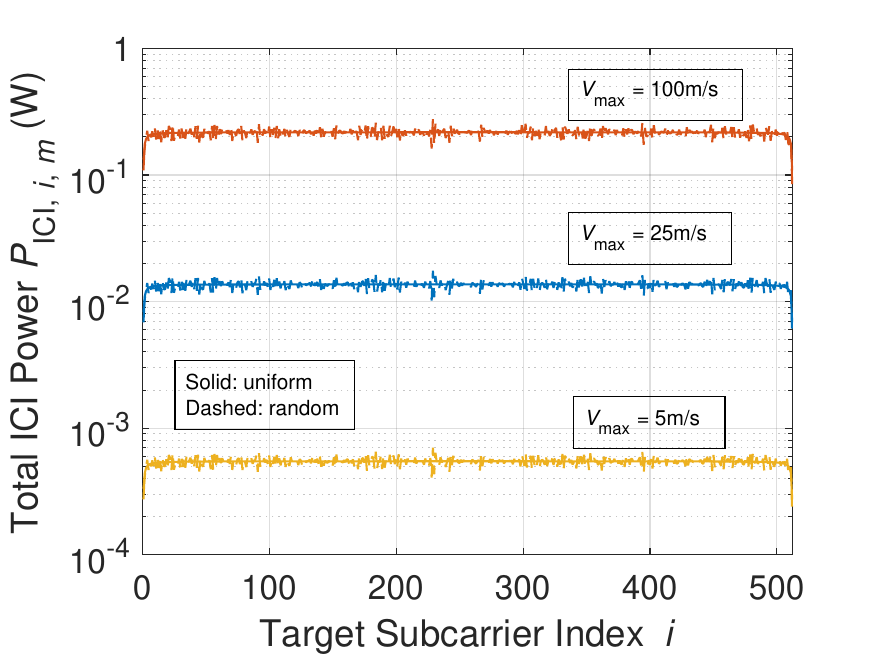}
        \label{fig:ici_distr_Nu8}
    }
    \subfigure[$N_{\mathcal{U}}=16$]
    {
        \includegraphics[width=3.0in]{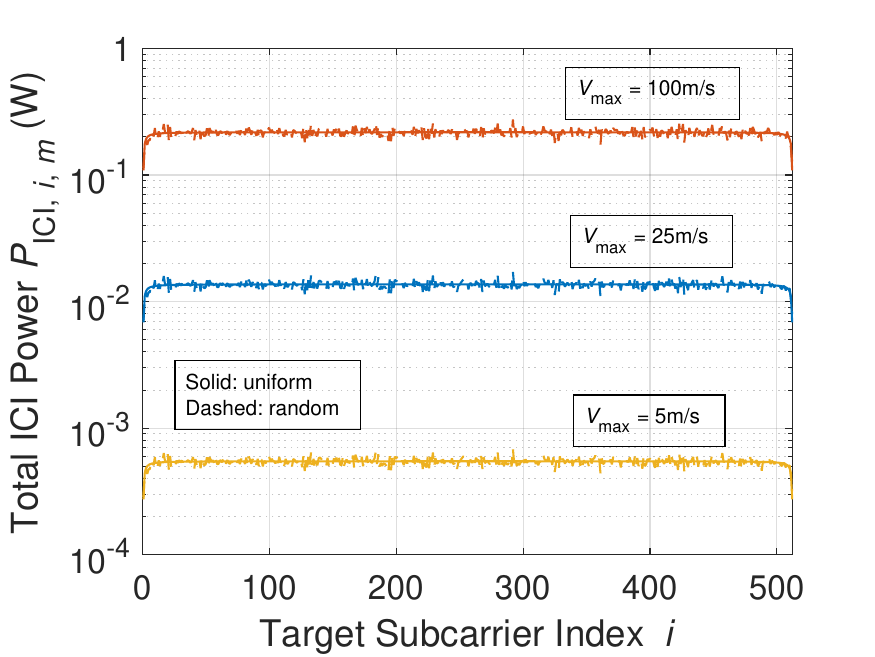}
        \label{fig:ici_distr_Nu16}
    }\\
    \subfigure[$N_{\mathcal{U}}=32$]
    {
        \includegraphics[width=3.0in]{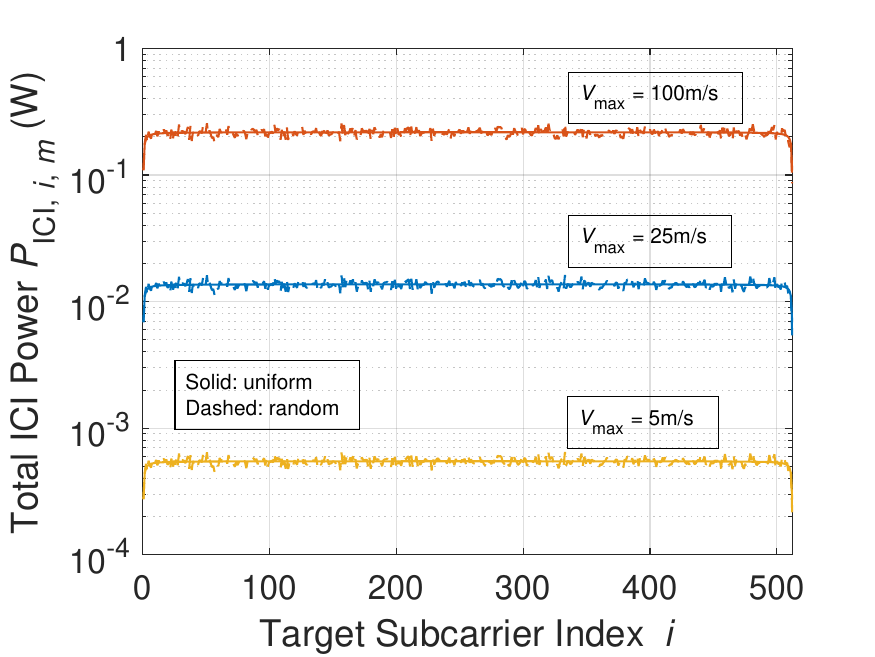}
        \label{fig:ici_distr_Nu32}
    }
    \subfigure[$N_{\mathcal{U}}=64$]
    {
        \includegraphics[width=3.0in]{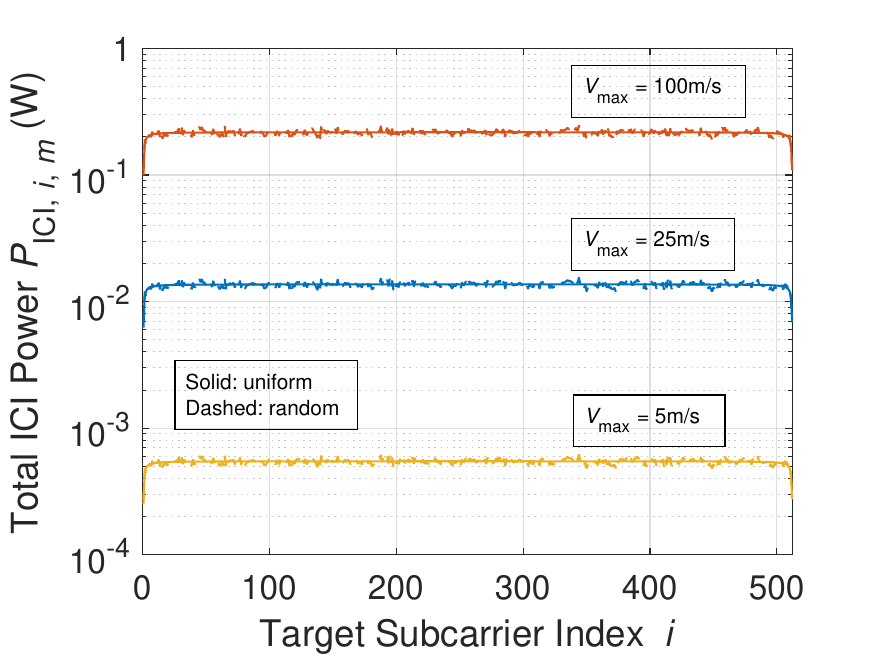}
        \label{fig:ici_distr_Nu64}
    }\\
    \caption{Variation of the ICI power across subcarrier frequencies for both uniform and random transmit power allocation. The dashed curves corresponding to random allocation fluctuate around the solid "deterministic" curves. Additionally, it can be observed that the variations die down as the value of $N_{\mathcal{U}}$ increases.}
    \label{fig:ici_distr}
\end{figure}

\begin{align}
    P_{\textrm {ICI}} =& \, \frac{P_{\textrm{T}}}{N_{\mathcal{C}}} \sum_{k=1}^{N_\mathcal{U}} \left( \sum_{j=1}^{N_\mathcal{G}} L_{i}( f_{j} ) - L_{i}( f_{i} ) \right) \label{eq:pici2} \\
    =& \, \frac{N_{\mathcal{U}}P_{\textrm{T}}}{N_{\mathcal{C}}} \left(1 - \frac {2}{\pi}\int _{0}^{\frac {\pi }{2}}\left [{ \frac {\textrm {Si}\left ({2 b\cos \left ({\psi }\right ) }\right )}{ b\cos \left ({\psi }\right )}-\frac {\sin ^{2}\left ({b\cos \left ({\psi }\right )}\right )}{\left ({b\cos \left ({\psi }\right )}\right )^{2}} }\right ]d\psi \right), \label{eq:pici3}
\end{align}
where $\textrm {Si}\left ({z}\right )\triangleq \int _{0}^{z}\frac {\sin \left ({t}\right )}{t}dt$ and $b\triangleq \frac {\pi V_{\textrm {max}}f_{\textrm{c}}T_{\textrm {s}}}{\textrm {c}}$. For a given $b$ (i.e., for fixed values of $V_{\textrm{max}}$, $f_{\textrm{c}}$ and $T_{\textrm {s}}$), it can be observed that the total ICI power is a function of the ratio $\frac{N_{\mathcal{U}}}{N_{\mathcal{C}}}$. But, we know from \eqref{eq:system-model} that $\frac{N_{\mathcal{U}}} {N_{\mathcal{C}}} = \frac{N_{\mathcal{R}}}{N_{\mathcal{G}}}$ is a constant. Thus, 
allowing multiple UEs on a given subcarrier or assigning multiple subcarriers to a UE may not always increase the total ICI incurred in the system. This can be seen in Figure~\ref{fig:ici_distr}, where the value on the y-axis corresponding to a given $V_{\textrm{max}}$ remains the same across all four subplots. We highlight here that the transition from \eqref{eq:pici2} to \eqref{eq:pici3} is reasonable only when the number of subcarriers is sufficiently large (e.g., $N_{\mathcal{G}} \geq 256$). 

For sufficiently small values of $b\triangleq \frac {\pi V_{\textrm {max}}f_{\textrm{c}}T_{\textrm {s}}}{\textrm {c}}$ i.e., $b \ll 1$, the $P_{\textrm{ICI}}$ in \eqref{eq:pici3} can be approximated as $P_{\textrm{ICI}} \approx  \frac{1}{18} \left(\frac {\pi V_{\textrm {max}}f_{\textrm{c}}T_{\textrm {s}}}{\textrm {c}}\right)^{2}\frac{N_{\mathcal {U}}P_{\textrm {T}}}{N_{\mathcal{C}}}$, which shows that for low values of $b$, the total ICI power scales quadratically in $V_{\textrm{max}}$. In addition, the ICI power is directly proportional to the square of the carrier frequency and inversely proportional to square of the subcarrier spacing (since $T_{\textrm{s}} = \frac{1}{\Delta f}$). Note that $b \ll 1$ is valid in scenarios of practical interest; for example, $b = 0.04$ when $V_{\textrm{max}} = 100$~km/h, $f_c = 2$~GHz, $T_s = 70~\mu$s as in an LTE system. Now, using the fact that that $P_{\textrm{ICI}}$ is the sum of individual ICI components across many statistically independent subcarriers and paths, and by applying the central limit theorem, it can be shown that the quantity $u_{i,m}[n]$ in \eqref{eq:rim} conforms to the complex Gaussian distribution $\mathcal{CN}\left(0, \frac{N_{\mathcal {U}}P_{\textrm {T}}}{N_{\mathcal{C}}}\sigma_{\textrm{u}}^{2}\right)$ where $\sigma_{\textrm{u}}^{2} = \frac{1}{18}\left(\frac {\pi V_{\textrm {max}}f_{\textrm{c}}T_{\textrm {s}}}{\textrm {c}}\right)^{2}$ denotes the normalized ICI power.

\section{Sum-Rate Performance}\label{sec:sumrateperformance}
In this section, we evaluate the sum-rate performance of the cellular system for two different receive combining schemes, zero-forcing (ZF) and maximal-ratio (MR) combining.
\subsection{Zero-Forcing Combining at the AP}
According to \eqref{eq:ri}, the $n$-th received data vector at the AP is given by
\begin{equation} \label{eq:rxzf1}
    \mathbf {r}_{i}\left [{n}\right ]=\sqrt {P_{\textrm {T}}}\mathbf {H}_{i}^{\mathcal {D}}\left [{n}\right ]\mathbf{D}_{\eta_{i}}^{1/2}\mathbf {x}_{i}\left [{n}\right ] + \mathbf {u}_{i}\left [{n}\right ] + \mathbf {n}_{i}\left [{n}\right].
\end{equation}
Since the AP does not know the actual channel, it is useful to express $\mathbf {H}_{i}^{\mathcal {D}}\left [{n}\right]$ in the above expression in terms of the channel estimate $\hat{\mathbf{H}}_{i}$ obtained via uplink training. Substituting \eqref{eq:ar2} and \eqref{eq:estimate2} in \eqref{eq:rxzf1}, we obtain
\begin{align} \label{eq:rxzf2}
    \mathbf {r}_{i}\left [{n}\right ] =& \, \sqrt {P_{\textrm {T}}} \, \hat{\mathbf{H}}_{i} \boldsymbol {\Lambda }_{i}\left [{n}\right] \mathbf{D}_{\eta_{i}}^{1/2}\mathbf {x}_{i}\left [{n}\right ] - \sqrt {P_{\textrm {T}}} \, \mathbf {G}_{i}^{\mathcal {P}} \boldsymbol {\Lambda }_{i}\left [{n}\right] \mathbf{D}_{\eta_{i}}^{1/2}\mathbf {x}_{i}\left [{n}\right ] \nonumber \\ \quad& + \, \sqrt {P_{\textrm {T}}} \, \mathbf {G}_{i}^{\mathcal {D}}\left [{n}\right] \mathbf{D}_{\eta_{i}}^{1/2}\mathbf {x}_{i}\left [{n}\right ] + \mathbf {u}_{i}\left [{n}\right ] + \mathbf {n}_{i}\left [{n}\right].
\end{align}
Thus, the received signal can now be interpreted as the desired signal having passed through a known channel and corrupted by additive interference plus noise terms. The second and  third terms in the above expression capture the multi-user interference that occur due to (i) ICI and noise entailed in channel estimation, and (ii) channel aging error, respectively. The AP then processes the received signal using zero-forcing combining. This involves pre-multiplying $\mathbf{r}_{i}\left[n\right]$ in \eqref{eq:rxzf2} by the pseudoinverse of $ \hat{\mathbf{H}}_{i}$. Thus, we have 
\begin{align}
    \mathbf {y}_{i}\left [{n}\right] =& \, \sqrt {P_{\textrm {T}}}\left(\hat{\mathbf{H}}_{i}\right)^{\dagger} \hat{\mathbf{H}}_{i} \boldsymbol {\Lambda }_{i}\left [{n}\right ]\mathbf{D}_{\eta_{i}}^{1/2}\mathbf {x}_{i}\left [{n}\right] - 
    \sqrt {P_{\textrm {T}}} \left (\hat{\mathbf{H}}_{i}\right)^{\dagger} \mathbf {G}_{i}^{\mathcal {P}} \boldsymbol {\Lambda }_{i}\left [{n}\right ]\mathbf{D}_{\eta_{i}}^{1/2}\mathbf {x}_{i}\left [{n}\right ] \nonumber \\ \quad& +
    \, \sqrt {P_{\textrm {T}}} \left (\hat{\mathbf{H}}_{i}\right )^{\dagger} \mathbf {G}_{i}^{\mathcal {D}}\left [{n}\right ]\mathbf{D}_{\eta_{i}}^{1/2}\mathbf {x}_{i}\left [{n}\right ] +
    \left (\hat{\mathbf{H}}_{i}\right )^{\dagger} \left(\mathbf {u}_{i}\left [{n}\right ] + \mathbf {n}_{i}\left [{n}\right ]\right). \label{eq:rxzf3}
\end{align}
The achievable uplink sum-rate of the above system is given by the following theorem. 
\begin{theorem}
The achievable uplink sum-rate of the cellular system when zero-forcing combining is employed at the AP is given by
\begin{equation}\label{eq:sumratezf}
    \textrm {C}_{\textrm{sum}}^{\textrm{zf,ul}}\left [{N_{\mathcal {D}}}\right ] = \sum_{i=1}^{N_{\mathcal{G}}} \sum_{k=1}^{N_{\mathcal{U}}} \textrm {C}^{\textrm{zf,ul}}_{\textrm {avg.}, ik}\left [{N_{\mathcal {D}}}\right],
\end{equation}
where $\textrm {C}^{\textrm{zf,ul}}_{\textrm {avg.}, ik}\left [{N_{\mathcal {D}}}\right]=\frac {1}{N_{\mathcal {P}}+N_{\mathcal {D}}} \sum _{n=1}^{N_{\mathcal {D}}}\textrm {C}^{\textrm{zf,ul}}_{ik}[n]$ denotes the average achievable uplink rate of UE$_{ik}$ on the $i$-th subcarrier across $N_{\mathcal{D}}$ transmissions, with $\textrm {C}^{\textrm{zf,ul}}_{ik}[n]$ denoting the achievable uplink rate of UE$_{ik}$'s $n$-th transmission on the $i$-th subcarrier. The quantity $\textrm {C}^{\textrm{zf,ul}}_{ik}[n]$ is given by
\begin{equation}\label{eq:ratezf}
    \textrm {C}^{\textrm{zf,ul}}_{ik}[n] \geq \Delta f \log_{2}\left(1 + \frac{(N_{\mathcal{B}} - N_{\mathcal{U}} + 1)\eta_{ik} \sigma _{\hat{\textrm {h}}}^{2} \bar {\lambda }\left [{n}\right ]}{\left(\frac{N_{\mathcal {U}}\sigma _{\textrm {u}}^{2}}{N_{\mathcal{C}}} + \frac{\sigma _{\textrm {n}}^{2}}{P_{\textrm {T}}}\right)\left(1+\frac{N_{\mathcal{V}}  \bar{\eta}_{i}}{N_{\mathcal{P}}} \bar {\lambda }\left [{n}\right ]\right) + \bar{\eta}_{i} \sigma _{\textrm {h}}^{2}\left ({1-\bar {\lambda }\left [{n}\right ]}\right )}\right),
\end{equation}
in which $\bar{\eta}_{i} = \sum_{k=1}^{N_{\mathcal{U}}} \eta_{ik}$,  ${N_{\mathcal{V}}  = \lceil{N_{\mathcal{C}}/N_{\mathcal{H}}}\rceil}$; $\bar {\lambda }\left [{n}\right ]=\frac {1}{V_{\textrm {max}}}\int _{0}^{V_{\textrm {max}}}J_{0}^{2}\left({\frac {2\pi vf_{\textrm{c}}nT_{\textrm {s}}}{\textrm {c}}}\right)dv$ denotes the expectation of the diagonal entries in $\left ({\boldsymbol {\Lambda }_{i}\left [{n}\right ]}\right )^{2}$, and $\sigma_{\hat{\textrm{h}}}^{2} = \sigma_{\textrm{h}}^{2} + \frac{N_{\mathcal{V}}  N_{\mathcal{U}}}{N_{\mathcal{P}} N_{\mathcal{C}}} \sigma_{\textrm{u}}^{2} + \frac{N_{\mathcal{V}} }{N_{\mathcal{P}}P_{\textrm{T}}} \sigma_{\textrm{n}}^{2}$ denotes the variance of the entries in $\hat{\mathbf{H}}_{i}$.
\end{theorem}

\begin{proof}
See Appendix~\ref{app:proof_theorem_1}.
\end{proof}

The proof follows by first computing the second-order statistics of the signal and interference terms over the randomness in the channels and user velocities, conditioned on the estimated channels. Then, a closed-form expression of the uplink sum-rate is obtained using a result from random matrix theory for the term that depends on the channel estimate. This results in a simple rate expression which agrees well with the numerical results, as we will see in the next section.

Since $\bar{\lambda}\left[{n}\right]$ is a decreasing function of $n$, the uplink rate and the sum-rate both decrease with the transmission index. This is expected, as higher transmission indices entail more channel aging. Further, at low values of $V_{\textrm{max}}$, the value of $\bar{\lambda}\left[{n}\right]$ is close to 1. As $V_{\textrm{max}}$ increases, $\bar{\lambda}\left[{n}\right]$ drops more sharply with $n$. Thus, the drop in the uplink rate across a given transmission frame becomes higher as the UEs move faster. From \eqref{eq:sumratezf} and \eqref{eq:ratezf}, we observe that the uplink sum-rate is a function of the power control coefficient $\eta_{ik}$. Recall that there are two constraints on $\eta_{ik}$. First, $0 \leq \eta_{ik} \leq 1$. Second, all power control coefficients of a UE must add up to one. With these constraints it is easy to show that, in the large antenna regime,  in order for the system to deliver maximum sum-rate performance each UE must distribute its transmit power uniformly among the $N_{\mathcal{C}}$ subcarriers assigned to it. Substituting $\eta_{ik} = \frac{1}{N_{\mathcal{C}}}$ in \eqref{eq:sumratezf}, we obtain the expression for the uplink sum-rate of the system as
\begin{equation}\label{eq:maxsumratezf}
    \textrm {C}_{\textrm{max,sum}}^{\textrm{zf,ul}}\left [{N_{\mathcal {D}}}\right ] = \frac{1}{N_{\mathcal{P}} + N_{\mathcal{D}}} \sum_{i=1}^{N_{\mathcal{G}}} \sum_{k=1}^{N_{\mathcal{U}}} \sum_{n=1}^{N_{\mathcal{D}}} \textrm {C}^{\textrm{zf,ul}}_{\textrm{max}, ik}[n],
\end{equation}
where $\textrm {C}^{\textrm{zf,ul}}_{\textrm{max}, ik}[n]$ denotes the uplink rate of UE$_{ik}$ on the $i$-th subcarrier and is given by
\begin{equation}\label{eq:maxratezf}
    \textrm {C}^{\textrm{zf,ul}}_{\textrm{max}, ik}[n] \geq \Delta f \log_{2}\left(1 + \frac{(N_{\mathcal{B}} - N_{\mathcal{U}} + 1) \sigma _{\hat{\textrm {h}}}^{2} \bar {\lambda }\left [{n}\right ]}{N_{\mathcal{C}} \left(\left(\frac{N_{\mathcal {U}}\sigma _{\textrm {u}}^{2}}{N_{\mathcal{C}}} + \frac{\sigma _{\textrm {n}}^{2}}{P_{\textrm {T}}}\right)\left(1+\frac{N_{\mathcal{V}}  N_{\mathcal{U}}}{N_{\mathcal{P}} N_{\mathcal{C}}} \bar {\lambda }\left [{n}\right ]\right) + \frac{N_{\mathcal{U}}}{N_{\mathcal{C}}} \sigma _{\textrm {h}}^{2}\left ({1-\bar {\lambda }\left [{n}\right ]}\right )\right)}\right).
\end{equation}

We highlight here that $\textrm {C}^{\textrm{zf,ul}}_{\textrm{max}, ik}[n]$ is the uplink rate of a UE on only one subcarrier. The total rate from a single user on the $n$-th transmission equals $N_{\mathcal{C}}\,\textrm {C}^{\textrm{zf,ul}}_{\textrm{max}, ik}[n]$. We note that the above expression compactly captures the dependence of the sum-rate on the various system parameters. 

Next, we consider the sum-rate performance with maximal-ratio combining at the AP. 

\subsection{Maximal-Ratio Combining at the AP}
With maximal-ratio combining at the AP, the processed signal $\mathbf {y}_{i}\left [{n}\right]$ is computed as $\mathbf {y}_{i}\left [{n}\right] = \left(\hat{\mathbf{H}}_{i}\right)^{H} \mathbf {r}_{i}\left [{n}\right]$. Thus, we have 

\begin{equation}\label{eq:yxmrc1}
\begin{aligned}
    \mathbf {y}_{i}\left [{n}\right] = \sqrt {P_{\textrm {T}}}\left(\hat{\mathbf{H}}_{i}\right)^{H} \hat{\mathbf{H}}_{i} \boldsymbol {\Lambda }_{i}\left [{n}\right ]\mathbf{D}_{\eta_{i}}^{1/2}\mathbf {x}_{i}\left [{n}\right] - 
    \sqrt {P_{\textrm {T}}} \left (\hat{\mathbf{H}}_{i}\right)^{H} \mathbf {G}_{i}^{\mathcal {P}} \boldsymbol {\Lambda }_{i}\left [{n}\right ]\mathbf{D}_{\eta_{i}}^{1/2}\mathbf {x}_{i}\left [{n}\right ] \\ +
    \, \sqrt {P_{\textrm {T}}} \left (\hat{\mathbf{H}}_{i}\right )^{H} \mathbf {G}_{i}^{\mathcal {D}}\left [{n}\right ]\mathbf{D}_{\eta_{i}}^{1/2}\mathbf {x}_{i}\left [{n}\right ] +
    \left (\hat{\mathbf{H}}_{i}\right )^{H} \left(\mathbf {u}_{i}\left [{n}\right ] + \mathbf {n}_{i}\left [{n}\right ]\right).
\end{aligned}
\end{equation}

The following theorem characterizes the achievable uplink sum-rate of the above system.
\begin{theorem}
The achievable uplink sum-rate of the cellular system when maximal-ratio combining is deployed at the AP is given by
\begin{equation}\label{eq:sumratemrc}
    \textrm {C}_{\textrm{sum}}^{\textrm{mrc,ul}}\left [{N_{\mathcal {D}}}\right ] = \sum_{i=1}^{N_{\mathcal{G}}} \sum_{k=1}^{N_{\mathcal{U}}} \textrm {C}^{\textrm{mrc,ul}}_{\textrm {avg.}, ik}\left [{N_{\mathcal {D}}}\right],
\end{equation}
where $\textrm {C}^{\textrm{mrc,ul}}_{\textrm {avg.}, ik}\left [{N_{\mathcal {D}}}\right]=\frac {1}{N_{\mathcal {P}}+N_{\mathcal {D}}} \sum _{n=1}^{N_{\mathcal {D}}}\textrm {C}^{\textrm{mrc,ul}}_{ik}[n]$ denotes the average achievable uplink rate of UE$_{ik}$ on the $i$-th subcarrier across $N_{\mathcal{D}}$ transmissions, with $\textrm {C}^{\textrm{mrc,ul}}_{ik}[n]$ denoting the achievable uplink rate of UE$_{ik}$'s $n$-th transmission on the $i$-th subcarrier. The quantity $\textrm {C}^{\textrm{mrc,ul}}_{ik}[n]$ is given by
\begin{equation}\label{eq:ratemrc}
    \textrm {C}^{\textrm{mrc,ul}}_{ik}[n] \geq \Delta f \log_{2}\left(1 + \frac{N_{\mathcal{B}} \eta_{ik} \sigma _{\hat{\textrm {h}}}^{2} \bar {\lambda }\left [{n}\right ]}{\left(\frac{N_{\mathcal {U}}\sigma _{\textrm {u}}^{2}}{N_{\mathcal{C}}} + \frac{\sigma _{\textrm {n}}^{2}}{P_{\textrm {T}}}\right)\left(1+\frac{N_{\mathcal{V}}  \bar{\eta}_{i}}{N_{\mathcal{P}}} \bar {\lambda }\left [{n}\right ]\right) + \bar{\eta}_{i} \sigma _{\textrm {h}}^{2}\left ({1-\bar {\lambda }\left [{n}\right ]}\right ) + \sigma _{\hat{\textrm {{h}}}}^{2}\, \bar{\eta}_{i}\, \bar {\lambda }\left [{n}\right ] }\right),
\end{equation}
in which $\bar{\eta}_{i} = \sum_{k=1}^{N_{\mathcal{U}}} \eta_{ik}$,  ${N_{\mathcal{V}}  = \lceil{N_{\mathcal{C}}/N_{\mathcal{H}}}\rceil}$; $\bar {\lambda }\left [{n}\right ]=\frac {1}{V_{\textrm {max}}}\int _{0}^{V_{\textrm {max}}}J_{0}^{2}\left({\frac {2\pi vf_{\textrm{c}}nT_{\textrm {s}}}{\textrm {c}}}\right)dv$ denotes the expectation of the diagonal entries in $\left ({\boldsymbol {\Lambda }_{i}\left [{n}\right ]}\right )^{2}$, and $\sigma_{\hat{\textrm{h}}}^{2} = \sigma_{\textrm{h}}^{2} + \frac{N_{\mathcal{V}}  N_{\mathcal{U}}}{N_{\mathcal{P}} N_{\mathcal{C}}} \sigma_{\textrm{u}}^{2} + \frac{N_{\mathcal{V}} }{N_{\mathcal{P}}P_{\textrm{T}}} \sigma_{\textrm{n}}^{2}$ denotes the variance of the entries in $\hat{\mathbf{H}}_{i}$.
\end{theorem}

\begin{proof}
See Appendix~\ref{app:proof_theorem_2}.
\end{proof}

For the same reason as in the previous subsection, the system delivers maximum sum-rate performance when each UE shares the transmit power equally among the $N_{\mathcal{C}}$ subcarriers allotted to it. In this case, the sum-rate when MRC is employed at the AP can be simplified as 
\begin{equation}\label{eq:maxsumratemrc}
    \textrm {C}_{\textrm{max,sum}}^{\textrm{mrc,ul}}\left [{N_{\mathcal {D}}}\right ] = \frac{1}{N_{\mathcal{P}}  + N_{\mathcal{D}}} \sum_{i=1}^{N_{\mathcal{G}}} \sum_{k=1}^{N_{\mathcal{U}}} \sum_{n=1}^{N_{\mathcal{D}}} \textrm {C}^{\textrm{mrc,ul}}_{\textrm{max}, ik}[n],
\end{equation}
where $\textrm {C}^{\textrm{mrc,ul}}_{\textrm{max}, ik}[n]$ denotes the achievable uplink rate of UE$_{ik}$ on the $i$-th subcarrier in its $n$-th transmission and is given by
\begin{equation}\label{eq:maxratemrc}
    \textrm {C}^{\textrm{mrc,ul}}_{\textrm{max}, ik}[n] \geq \Delta f \log_{2}\left(1 \! + \! \frac{N_{\mathcal{B}}\,\sigma _{\hat{\textrm {h}}}^{2} \, \bar {\lambda }\left [{n}\right ]}{N_{\mathcal{C}}\left(\left(\frac{N_{\mathcal {U}}\sigma _{\textrm {u}}^{2}}{N_{\mathcal{C}}} \! +\! \frac{\sigma _{\textrm {n}}^{2}}{P_{\textrm {T}}}\right)\left(1\! + \! \frac{N_{\mathcal{V}}  N_{\mathcal{U}}}{N_{\mathcal{P}}N_{\mathcal{C}}} \bar {\lambda }\left [{n}\right ]\right) \! + \! \frac{N_{\mathcal{U}}}{N_{\mathcal{C}}} \sigma_{\textrm {h}}^{2}\left ({1 \! - \! \bar {\lambda }\left [{n}\right ]}\right ) \! + \! \frac{N_{\mathcal{U}}}{N_{\mathcal{C}}}\sigma _{\hat{\textrm {{h}}}}^{2} \, \bar {\lambda }\left [{n}\right]\right)}\right).
\end{equation}
The uplink rate of a UE equals $N_{\mathcal{C}}\,\textrm {C}^{\textrm{mrc,ul}}_{\textrm{max}, ik}[n]$ since each UE transmits on $N_{\mathcal{C}}$ subcarriers.

Comparing \eqref{eq:maxratezf} with \eqref{eq:maxratemrc}, we observe that the numerator of the uplink SINRs for the ZF and the MRC receiver are approximately the same when $N_{\mathcal{B}} \gg N_{\mathcal{U}}$. However, the presence of an additional $\frac{N_{\mathcal{U}}}{N_{\mathcal{C}}}\sigma _{\hat{\textrm {{h}}}}^{2} \, \bar {\lambda }\left [{n}\right]$ in the denominator of the SINR in \eqref{eq:maxratemrc}, which arises due to the fact that MRC only maximizes the desired signal power and ignores interference, will cause the MRC receiver to achieve lower sum-rate than the ZF receiver. The expressions also capture the effect of channel estimation errors, multi-user interference due to imperfect beamforming, channel aging, and inter-carrier interference due to the Doppler shifts caused by user mobility. For example, at a fixed $V_{\max}$ value, as the transmit symbol index $n$ increases, $\bar{\lambda}[n]$ decreases, thereby reducing the achievable uplink rate. Similarly, as pilot/data transmit power increases, the achievable sum-rate increases, but the interference power also increases with the transmit power, leading to a saturation of the sum-rate. Similar inferences can be drawn about the dependence of the sum-rate on other system parameters; these are illustrated in the next section, where we present numerical results. 

From \eqref{eq:maxratezf} and \eqref{eq:maxratemrc}, it is evident that the uplink rate per user on the $i$-th subcarrier increases with the subcarrier spacing $\Delta f$, because of the linear dependence of $\textrm {C}^{\textrm{mrc,ul}}_{\textrm{max}, ik}[n]$ and $\textrm {C}^{\textrm{zf,ul}}_{\textrm{max}, ik}[n]$ on $\Delta f$ and the inverse dependence of the ICI power on $\Delta f$. However, if the total bandwidth, $B$, is fixed, the total number of subcarriers will go down, i.e., the value of $N_{\mathcal{G}} = B/\Delta f$ in \eqref{eq:maxratemrc} and \eqref{eq:maxratezf} will be smaller. Deeper understanding the impact of increasing $\Delta f$ requires one to revisit the subcarrier and symbol (resource block) allocation scheme as the number of subcarriers is varied, and is an interesting direction for future work. In this paper, we  assume that $\Delta f$ and $B$ are given and fixed.

\section{Numerical Results}\label{sec:numericalresults}
We consider an AP equipped with  $N_{\mathcal{B}} = 256$ antennas that serves $N_{\mathcal{R}} = 2048$ single antenna users. In all, there are $N_{\mathcal{G}} = 512$ subcarriers spaced $\Delta f = 10\,\textrm{kHz}$ apart and spanning a total bandwidth $B = 5.12\,\textrm{MHz}$. The carrier center frequency is set at $f_{\textrm{c}} = 3\,\textrm{GHz}$. 
We employ path loss inversion based power control and set the effective transmit power $P_{\textrm{T}}$ at the UEs such that SNR at the AP is $10\,\textrm{dB}$. We assume the coherence bandwidth  to be $B_{\textrm{c}}=300\,\textrm{kHz}$ \cite{marzetta_book_2016}. Thus, there are $N_{\mathcal{H}}=30$ subcarriers in one coherence bandwidth. To observe the benefits of multiple access, the number of UEs served by a subcarrier is varied from $N_{\mathcal{U}} = 4$ to $N_{\mathcal{U}} = 128$ in powers of two such that the ratio $\frac{N_{\mathcal{R}}}{N_{\mathcal{U}}}$ is an integer. Correspondingly, the number of subcarriers assigned to a UE varies from $N_{\mathcal{C}} = 1$ to $N_{\mathcal{C}} = 32$, such that the ratio $\frac{N_{\mathcal{U}}}{N_{\mathcal{C}}}$ remains constant. We note that  $N_{\mathcal{U}} =  N_{\mathcal{R}}/N_{\mathcal{G}} = 4$ and $N_{\mathcal{C}} = 1$ corresponds to the allocation adopted in \cite{zhang_tcom_2017} where each subcarrier serves a non-overlapping set of 4 UEs while each UE transmits on a single subcarrier.

Figure \ref{fig:rate1} shows a plot of the achievable uplink rate of a UE as a function of the transmission index up to 30 transmissions for $V_{\textrm{max}} = 5$,  $25$ and $100 \, \textrm{m/s}$. The results for zero-forcing and maximal-ratio combining are shown in separate sub-figures. In both figures, we assume $N_{\mathcal{U}} = 8$ users per subcarrier. It can be observed that the simulated and the analytical curves match closely, thereby establishing the accuracy of our analytical uplink rate expressions. The uplink rate values for the ZF receiver for any given value of $V_{\textrm{max}}$ are significantly higher than those for the MRC receiver. This can be attributed to the 
presence of the additional $\frac{N_{\mathcal{U}}}{N_{\mathcal{C}}}\sigma _{\hat{\textrm {{h}}}}^{2} \, \bar {\lambda }\left [{n}\right]$  in the denominator of the SINR expression in \eqref{eq:maxratemrc}, and the fact that at $P_T = 10$ dB, the system is highly interference-limited. As a consequence, the interference suppression by the ZF receiver significantly improves the rate, in spite of the channel estimates getting outdated over time. The term $\frac{N_{\mathcal{U}}}{N_{\mathcal{C}}}\sigma _{\hat{\textrm {{h}}}}^{2} \, \bar {\lambda }\left [{n}\right]$ is also responsible for the relatively slower drop in the uplink rate for the MRC receiver compared to the ZF receiver.

\begin{figure}
    \centering
    \subfigure[ZF Receiver]
    {
        \includegraphics[width=3.0in]{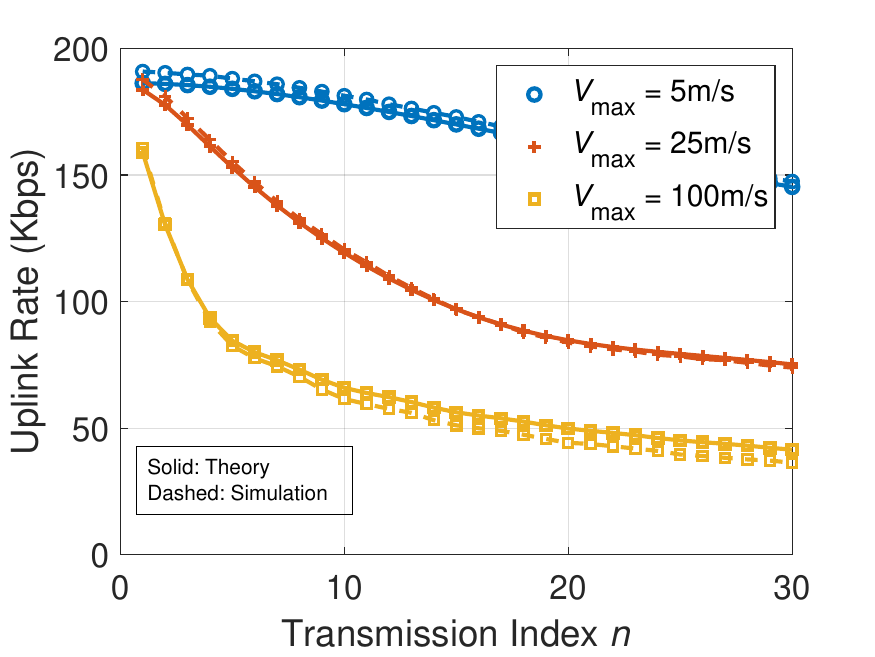}
        \label{fig:ratezf_Nu8}
    }    
    \subfigure[MRC Receiver]
    {
        \includegraphics[width=3.0in]{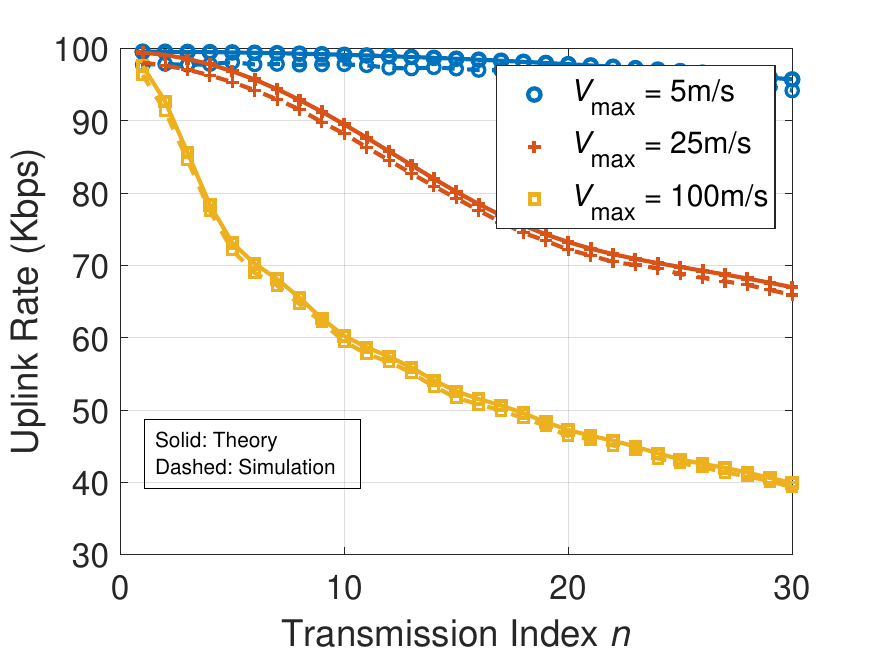}
        \label{fig:ratemrc_Nu8}
    }
    \caption{Uplink rate as a function of  transmission index for different maximum user velocities, for both ZF and MRC receivers. The ZF receiver significantly outperforms MRC at low transmission indices, where the channel estimates are reasonably up-to-date. For later transmission indices, and at high $V_{\text{max}}$, the performance of the two receivers become nearly equal.}
    \label{fig:rate1}
\end{figure}

Figure \ref{fig:rate2} shows the analytical plots of the achievable uplink rate of a UE with zero-forcing combining at the AP for all possible values of $N_{\mathcal{U}} < N_{\mathcal{B}}$. The sub-figures \ref{fig:ratezf_v5}, \ref{fig:ratezf_v25} and \ref{fig:ratezf_v100} correspond to $V_{\textrm{max}} = 5$,  $25$ and  $100\,\textrm{m/s}$, respectively.  It can be observed from the three plots that having each subcarrier serve more number of UEs generally results in higher uplink rate per UE. Depending on the value of $V_{\textrm{max}}$, such gains in UE performance may last up to $N_{\mathcal{U}} = 128$ (as shown for $V_{\textrm{max}} = 5\,\textrm{m/s}$ in figure \ref{fig:ratezf_v5}) or saturate at a lower $N_{\mathcal{U}}$ value (as shown in figure \ref{fig:ratezf_v25} and \ref{fig:ratezf_v100}). In particular, for $V_{\textrm{max}} = 100$ m/s, $N_{\mathcal{U}} = 64$ provides a better uplink rate compared to $N_{\mathcal{U}}=128$ beyond a transmission index of $5$, because the outdated channel estimate fails to suppress interference effectively. Thus, scheduling too many users per subcarrier is not advisable at high $V_{\max}$ when a ZF receiver is employed. On the other hand, for the MRC receiver, the uplink sum-rate monotonically increases with $N_{\mathcal{U}}$ for all three values of $V_{\textrm{max}}$ within the range of transmission indices considered, as depicted in Figure~\ref{fig:rate3}. Also, the uplink sum-rate achieved with the MRC receiver is significantly lower than that with the ZF receiver at lower transmission indices and maximum user velocities, while MRC outperforms ZF at higher transmission indices and maximum user velocities. This shows that the MRC receiver is more robust to channel aging than the ZF receiver under harsh channel conditions.

\begin{figure}
    \centering
    \subfigure[$V_{\textrm{max}} = 5\,\textrm{m/s}$]
    {
        \includegraphics[width=2.25in]{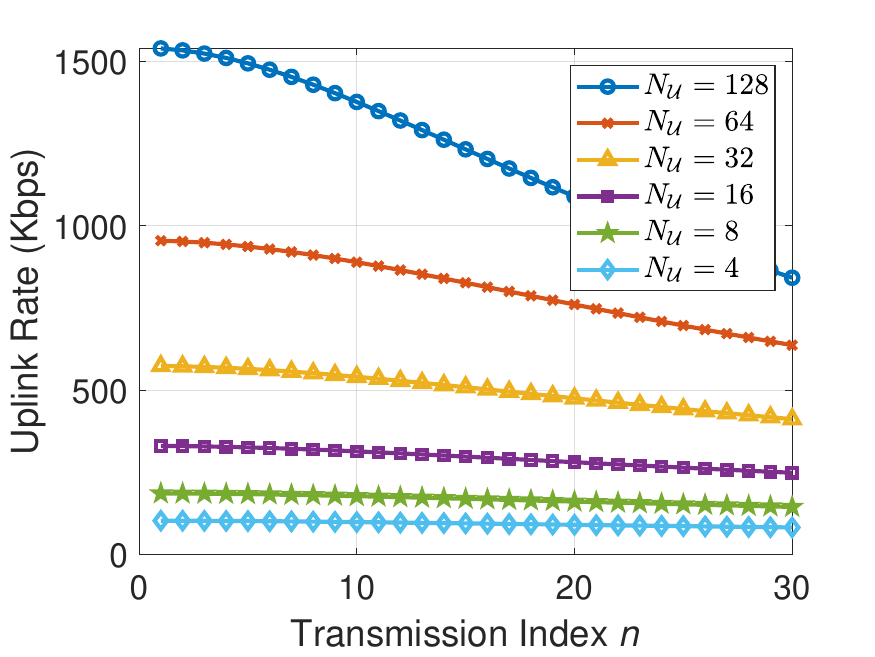}
        \label{fig:ratezf_v5}
    }\hspace{-0.8cm}
    \subfigure[$V_{\textrm{max}} = 25\,\textrm{m/s}$]
    {
        \includegraphics[width=2.2in]{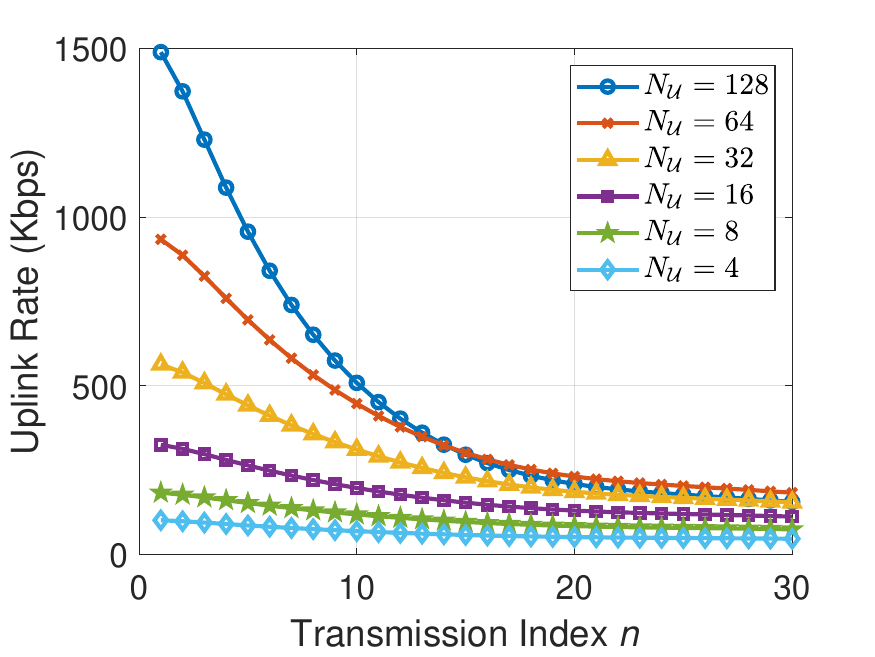}
        \label{fig:ratezf_v25}
    }\hspace{-0.8cm}
    \subfigure[$V_{\textrm{max}} = 100\,\textrm{m/s}$]
    {
        \includegraphics[width=2.2in]{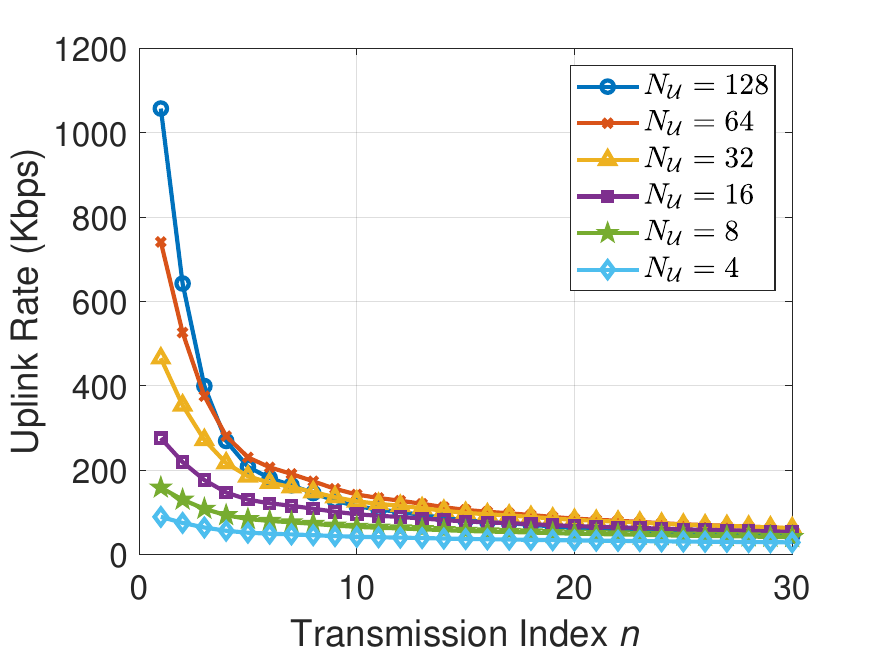}
        \label{fig:ratezf_v100}
    }
    \caption{Plot of uplink rate \emph{vs.} transmission index for ZF receiver across different values of $N_{\mathcal{U}}$. When $V_{\textrm{max}} = 5$~m/s, the uplink rate monotonously increases till $N_{\mathcal{U}} = 128$. However, at higher values of $V_{\textrm{max}}$, it is $N_{\mathcal{U}} = 64$ that consistently delivers the best per-UE performance throughout the transmission period.}
    \label{fig:rate2}
\end{figure}

\begin{figure}
    \centering
    \subfigure[$V_{\textrm{max}} = 5\,\textrm{m/s}$]
    {
        \includegraphics[width=2.25in]{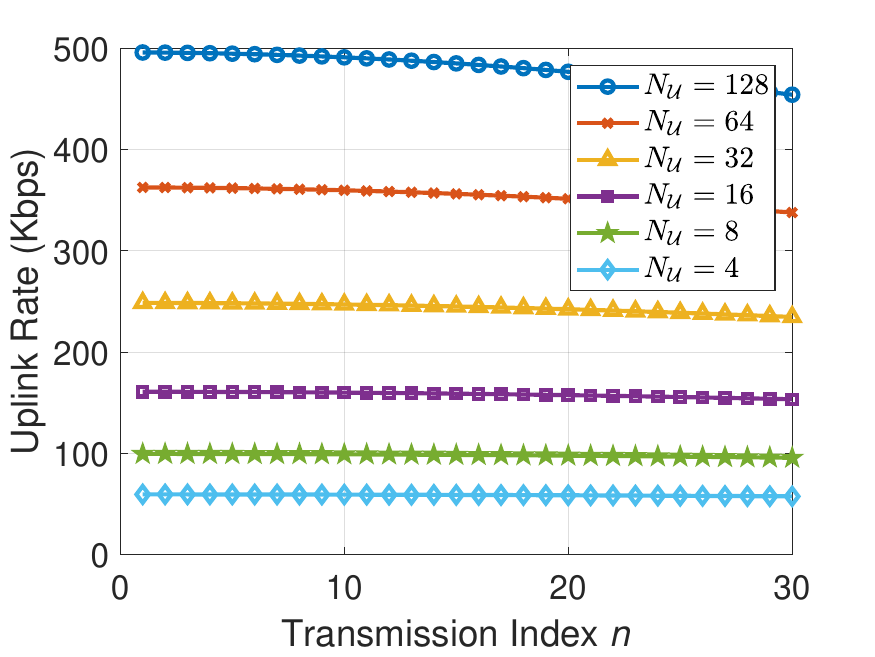}
        \label{fig:ratemrc_v5}
    }\hspace{-0.8cm}
    \subfigure[$V_{\textrm{max}} = 25\,\textrm{m/s}$]
    {
        \includegraphics[width=2.2in]{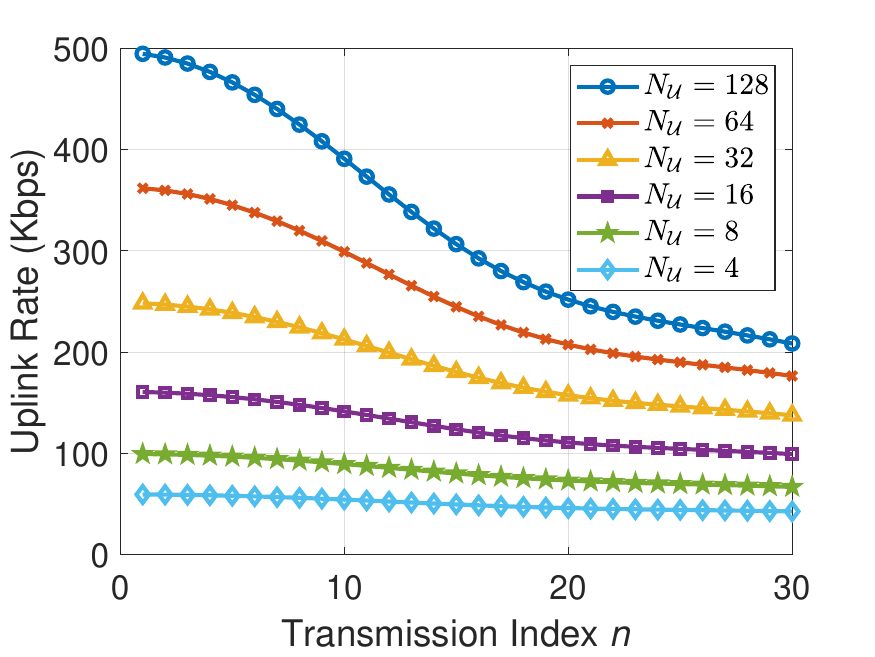}
        \label{fig:ratemrc_v25}
    }\hspace{-0.8cm}
    \subfigure[$V_{\textrm{max}} = 100\,\textrm{m/s}$]
    {
        \includegraphics[width=2.2in]{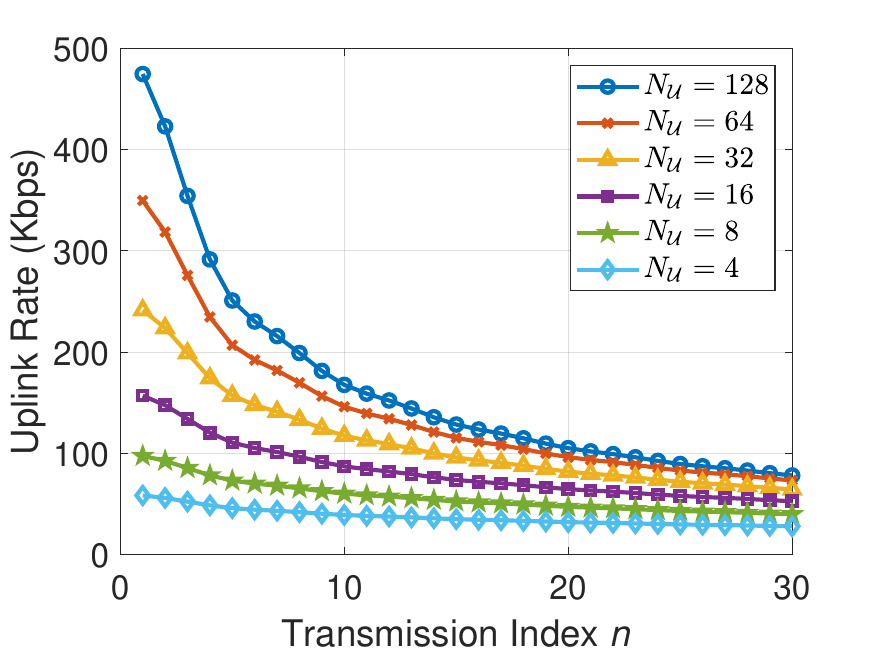}
        \label{fig:ratemrc_v100}
    }
    \caption{Plot of uplink rate \emph{vs.} transmission index for MRC receiver across different values of $N_{\mathcal{U}}$. Unlike ZF, with MRC receiver, the uplink rate attains its maximum value at $N_{\mathcal{U}} = 128$ for all three values of $V_{\textrm{max}}$.}
    \label{fig:rate3}
\end{figure}

Figure \ref{fig:sumrate1} demonstrates the achievable sum-rate performance of the cellular system at $N_{\mathcal{U}} = 8$ as a function of $V_{\textrm{max}}$ for three different values of the pilot percentage. As noted in \eqref{eq:maxsumratezf} and \eqref{eq:maxsumratemrc}, the system sum-rate is obtained by adding the rates of all individual users across $N_{\mathcal{D}}$ transmissions and dividing the result by the sum of the pilot and data lengths, i.e., $N_{\mathcal{P}} + N_{\mathcal{D}}$. The pilot percentage, defined as $\mu = \frac{N_{\mathcal{P}}}{N_{\mathcal{P}} + N_{\mathcal{D}}} \times 100$, is varied by keeping the length of the pilot sequence $N_{\mathcal{P}}$ fixed and changing the transmission length $N_{\mathcal{D}}$. We observe that the simulated curves corroborate the analytical plots well. With ZF combining at the AP, the system delivers the best sum-rate performance for $5 \leq V_{\textrm{max}}\, \leq 10\,\textrm{m/s}$ when $\mu$ equals 12.5\%. For $V_{\textrm{max}}$ beyond $10\,\textrm{m/s}$, $\mu = 25\%$ gives the best sum-rate performance. In other words, a shorter frame duration is better when the channel varies faster over time. With MRC at the AP, the behavior is similar except that the change in the optimal pilot percentage occurs at around $V_{\max} = 38\,\textrm{m/s}$. 

\begin{figure}
    \centering
    \subfigure[ZF Receiver]
    {
        \includegraphics[width=3.0in]{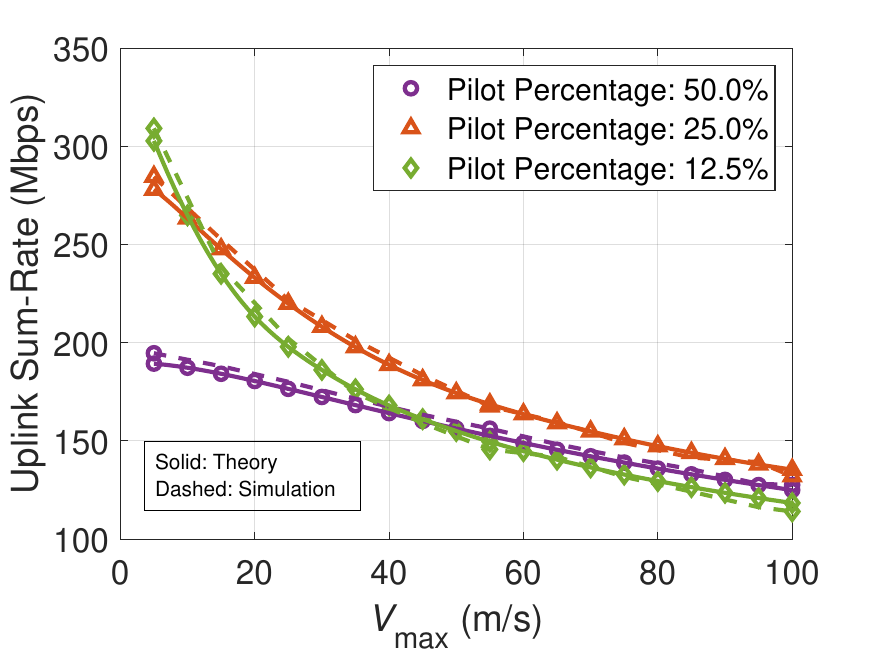}
        \label{fig:sumratezf_Nu8}
    }
    \subfigure[MRC Receiver]
    {
        \includegraphics[width=3.0in]{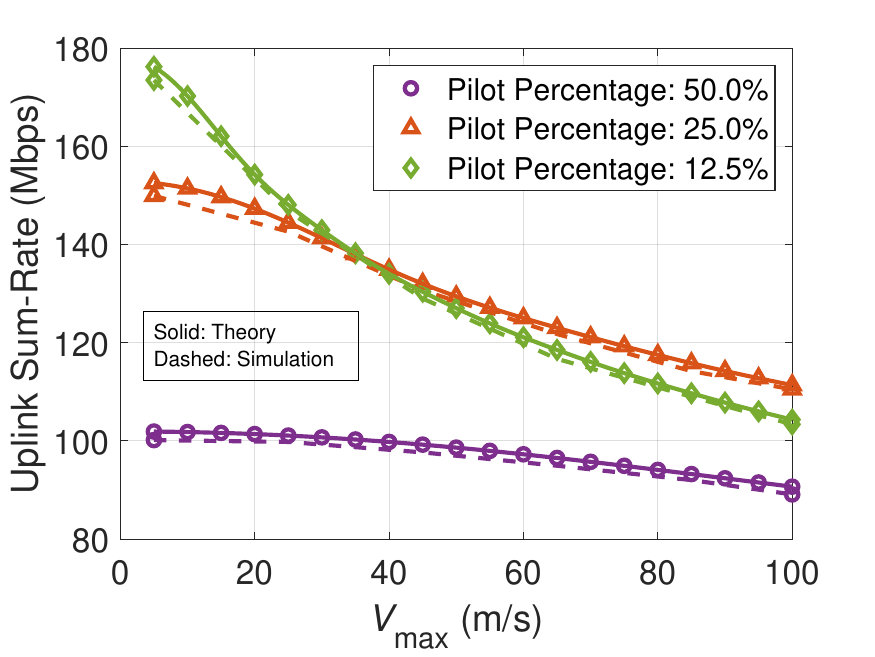}
        \label{fig:sumratemrc_Nu8}
    }
    \caption{Plots of system sum-rate \emph{vs.} $V_{\textrm{max}}$ for ZF and MRC receivers. To ensure maximum sum-rate performance, the system employing ZF receiver must switch the pilot percentage from 12.5\% to 25\% when $V_{\textrm{max}}$ crosses $10$~m/s. With MRC receiver, the change in the optimal pilot percentage should occur at around $V_{\textrm{max}} =  38$~ m/s.}
    \label{fig:sumrate1}
\end{figure}

Figure \ref{fig:sumrate2} shows the analytical plots of the sum-rate with ZF combining at the AP for the different possible values of $N_{\mathcal{U}}$. The sub-figures \ref{fig:sumratezf_pp125}, \ref{fig:sumratezf_pp25} and \ref{fig:sumratezf_pp50} correspond to $\mu = 12.5\%$, $\mu = 25\%$ and $\mu = 50\%$ respectively. It can be observed that increasing the number of UEs served by a subcarrier generally increases the sum-rate performance of the cellular system. However, such gains in the system sum-rate occur only up to $N_{\mathcal{U}} = 64$, after which, the sum-rate decreases. The behaviour is similar with MRC as can be seen in Figure~\ref{fig:sumrate3}.

\begin{figure}
    \centering
    \subfigure[Pilot percentage: $12.5\%$]
    {
        \includegraphics[width=2.25in]{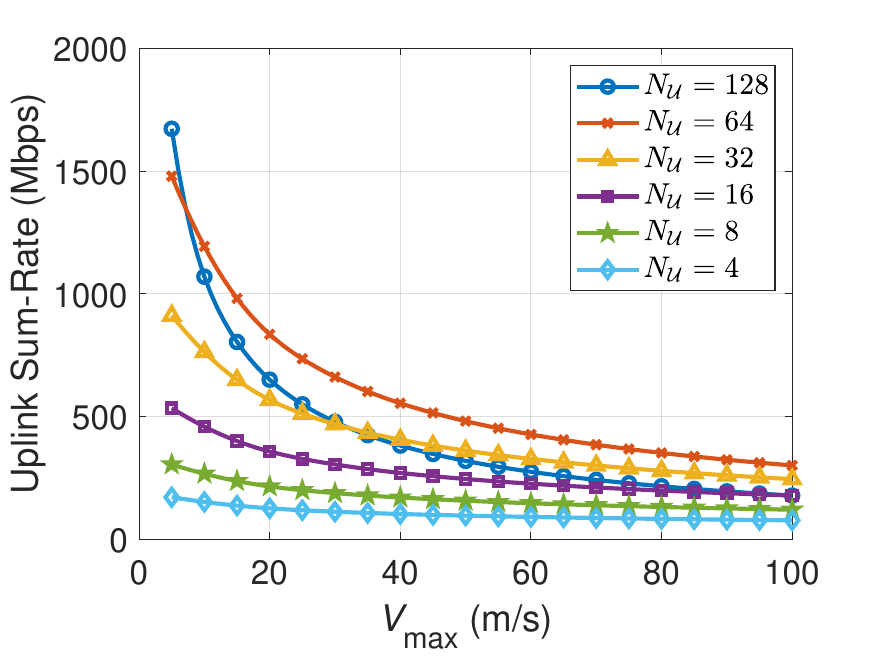}
        \label{fig:sumratezf_pp125}
    }\hspace{-0.8cm}
    \subfigure[Pilot percentage: $25.0\%$]
    {
        \includegraphics[width=2.2in]{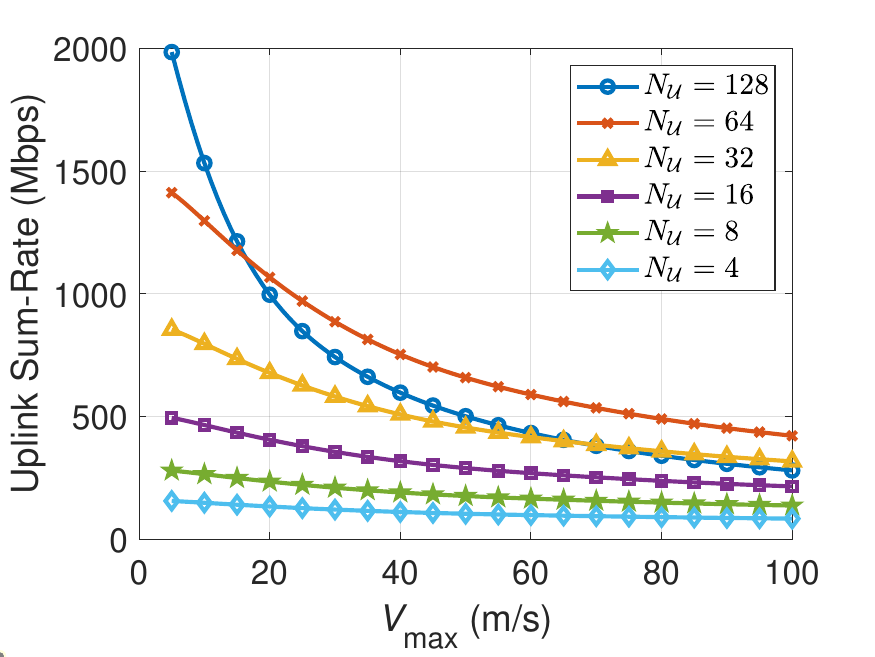}
        \label{fig:sumratezf_pp25}
    }\hspace{-0.8cm}
    \subfigure[Pilot percentage: $50.0\%$]
    {
        \includegraphics[width=2.2in]{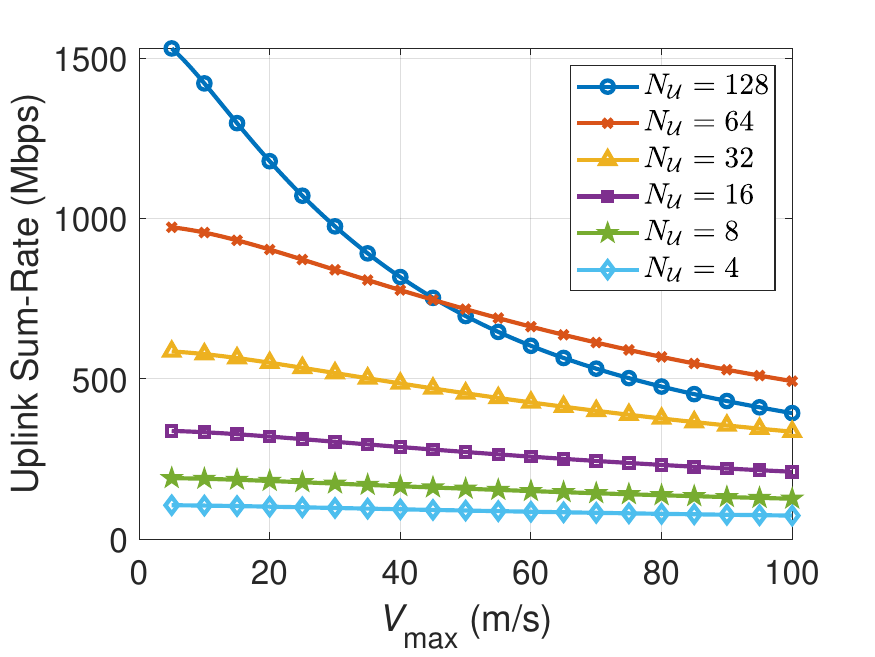}
        \label{fig:sumratezf_pp50}
    }
    \caption{Plots of system sum-rate \emph{vs.} $V_\textrm{max}$ for ZF receiver across different values of $N_{\mathcal{U}}$. Performance gains obtained by increasing $N_{\mathcal{U}}$ saturate at $N_{\mathcal{U}} = 64$.}
    \label{fig:sumrate2}
\end{figure}

\begin{figure}
    \centering
    \subfigure[Pilot percentage: $12.5\%$]
    {
        \includegraphics[width=2.2in]{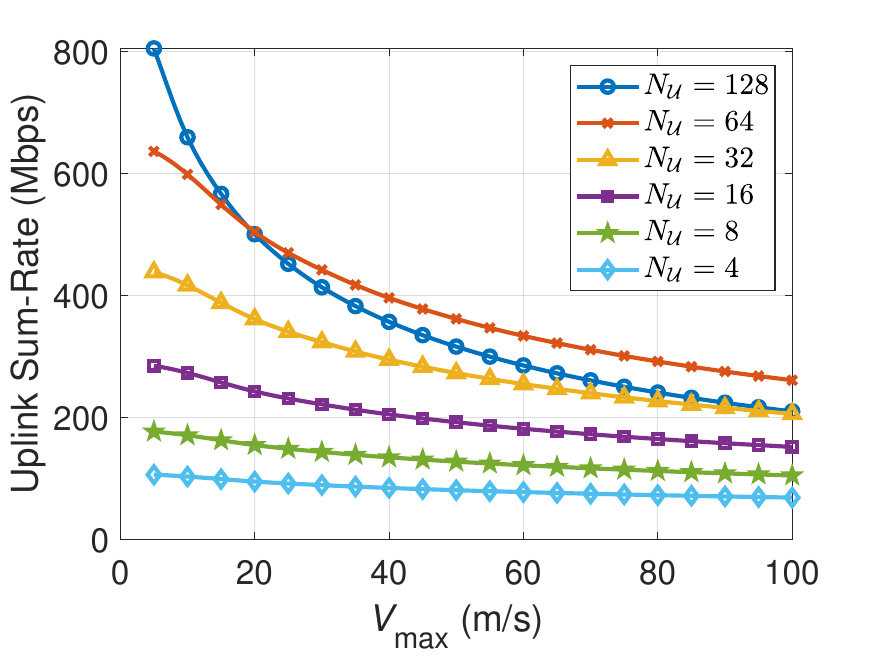}
        \label{fig:sumratemrc_pp125}
    }\hspace{-0.8cm}
    \subfigure[Pilot percentage: $25.0\%$]
    {
        \includegraphics[width=2.25in]{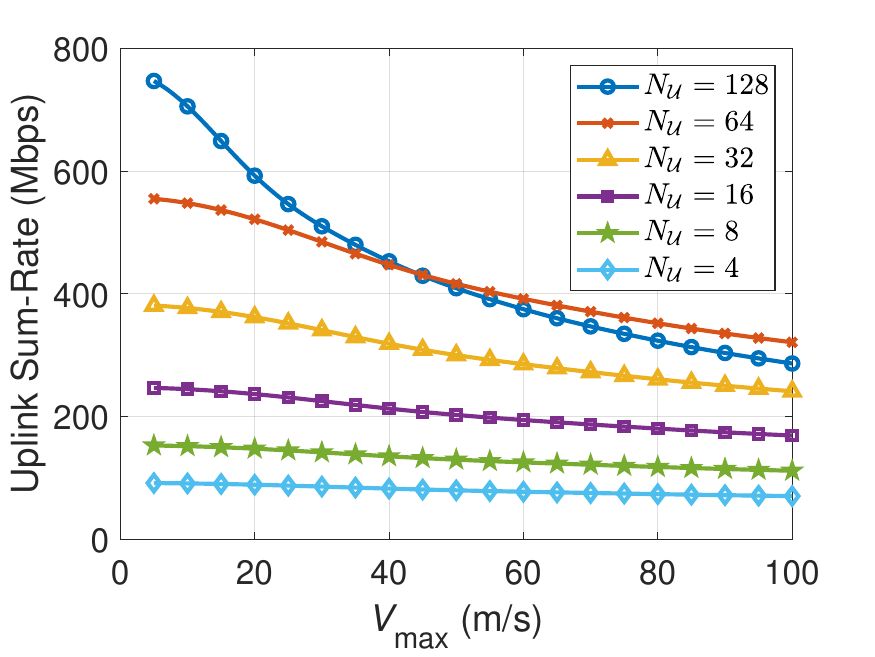}
        \label{fig:sumratemrc_pp25}
    }\hspace{-0.8cm}
    \subfigure[Pilot percentage: $50.0\%$]
    {
        \includegraphics[width=2.2in]{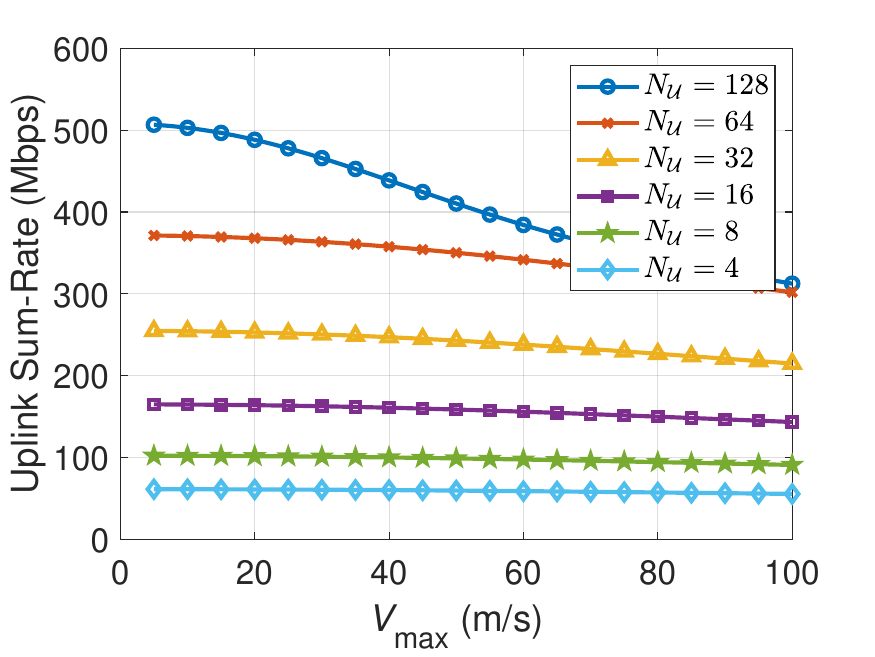}
        \label{fig:sumratemrc_pp50}
    }
    \caption{Plots of system sum-rate \emph{vs.} $V_\textrm{max}$ for MRC receiver across different values of $N_{\mathcal{U}}$. Similar to ZF, the MRC receiver delivers the best sum-rate performance at $N_{\mathcal{U}} = 64$.}
    \label{fig:sumrate3}
\end{figure}

Figure \ref{fig:wavy_sum_rate} shows plots of sum-rate as a function of the transmission length $N_{\mathcal{D}}$ for all possible values of $N_{\mathcal{U}}$, $V_{\textrm{max}} = 5, 25, 100\,\textrm{m/s}$, and both ZF and MRC receivers. From the figure, we can see that it is easy to determine the value of $N_{\mathcal{D}}$ which achieves the optimal trade-off between the pilot overhead and the effect of channel aging. As expected, the optimal frame duration is lower for higher $V_{\textrm{max}}$. Further, at $V_{\textrm{max}} = 100$~m/s, and with both ZF and MRC, the optimal uplink sum-rate with $N_{\mathcal{U}} = 64$ exceeds that with $N_{\mathcal{U}} = 128$. Thus, when the channel is very fast-varying, it is better to allot fewer users per subcarrier. In Figure \ref{fig:opt_sum_rate}, we plot the sum-rate optimized over $N_{\mathcal{D}}$ as a function of the received SNR for $N_{\mathcal{U}} = 16, 32, 64, 128$, $V_{\textrm{max}} = 5, 25, 100$~m/s, and both ZF and MRC receivers. At low SNR, the performance of ZF and MRC are close to each other. Since the MRC receiver is computationally simpler than the ZF, it is preferable at low SNRs. Also, while MRC outperforms ZF at low SNR when $N_{\mathcal{U}} = 128$, the performance obtained by choosing the value of $N_{\mathcal{U}}$ that yields the best sum-rate at each SNR point with ZF exceeds that obtained from MRC, even at low SNRs. Thus, interference suppression via ZF is useful even at low SNRs, when multiple mobile users are scheduled on each subcarrier.

\begin{figure}
    \centering
    \subfigure[$V_{\textrm{max}} = 5\,\textrm{m/s}$]
    {
        \includegraphics[width=2.2in]{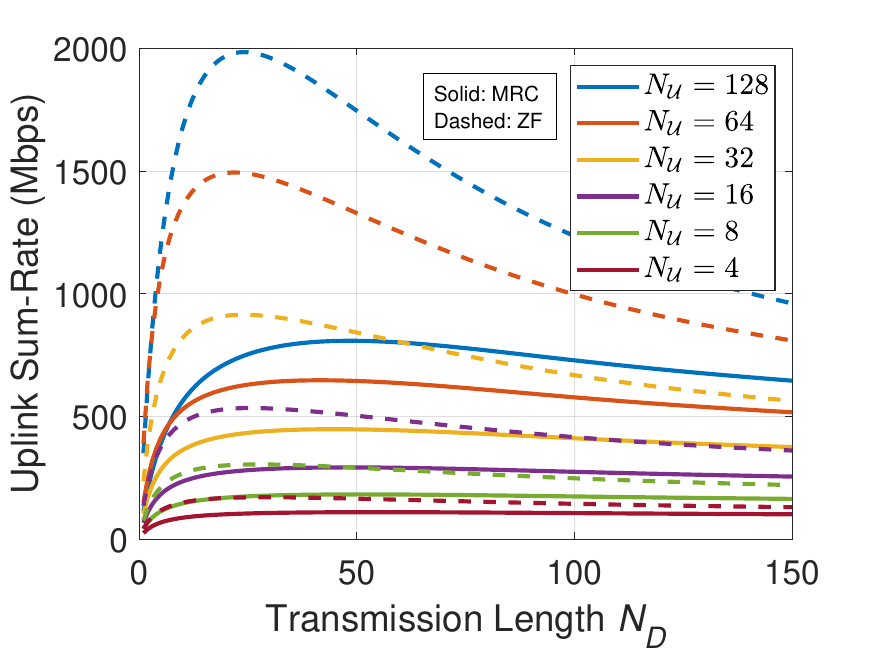}
        \label{fig:wavy_sum_rate_vmax5}
    }\hspace{-0.8cm}
    \subfigure[$V_{\textrm{max}} = 25\,\textrm{m/s}$]
    {
        \includegraphics[width=2.2in]{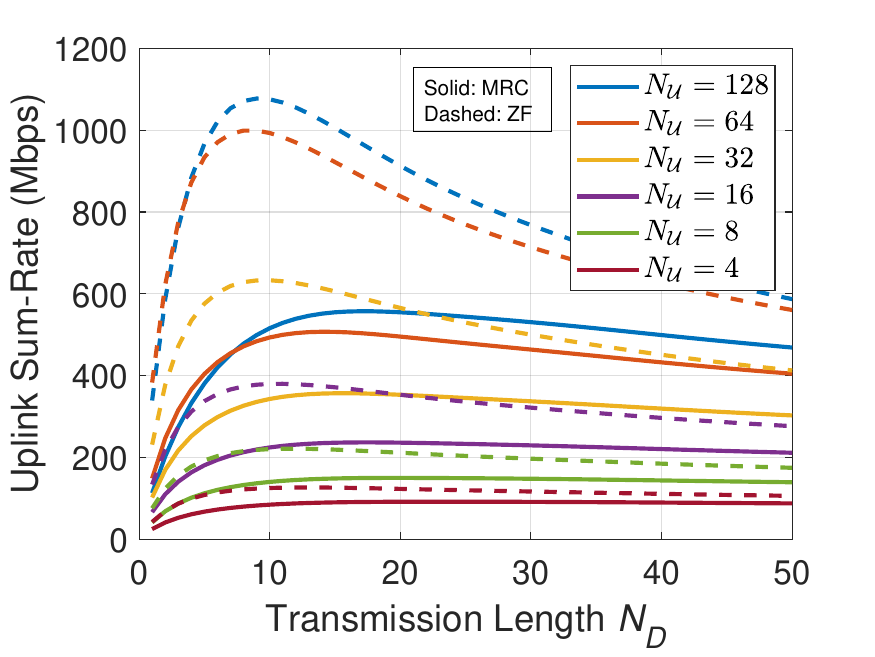}
        \label{wavy_sum_rate_vmax25}
    }\hspace{-0.8cm}
    \subfigure[$V_{\textrm{max}} = 100\,\textrm{m/s}$]
    {
        \includegraphics[width=2.25in]{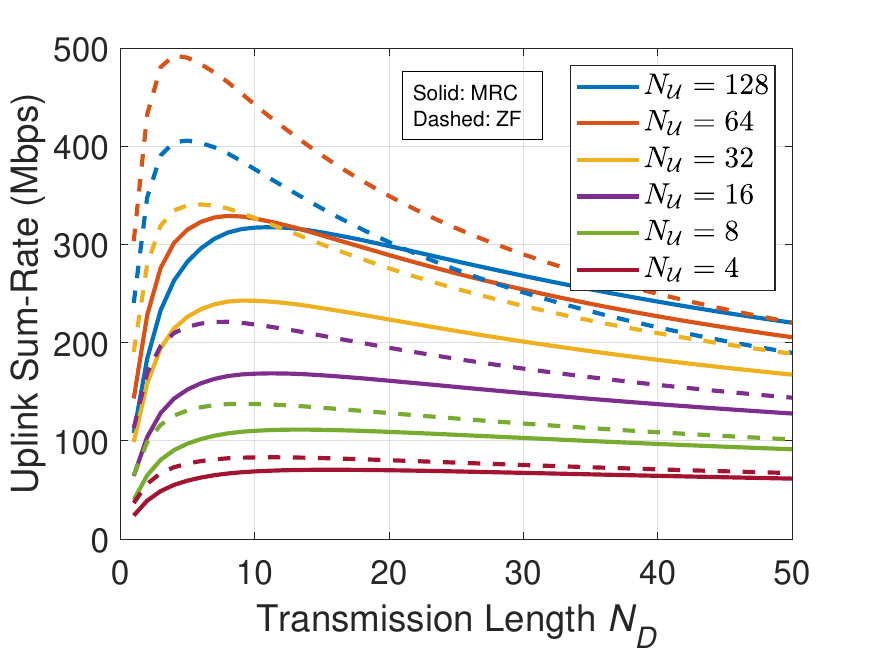}
        \label{fig:wavy_sum_rate_vmax100}
    }
    \caption{Plots of system sum-rate \emph{vs.} transmission length for both receivers. With $N_{\mathcal{D}}$ and $N_{\mathcal{U}}$ chosen optimally, ZF outperforms MRC in all cases, at an SNR of $10$~dB.}
    \label{fig:wavy_sum_rate}
\end{figure}

\begin{figure}
    \centering
    \subfigure[$V_{\textrm{max}} = 5\,\textrm{m/s}$]
    {
        \includegraphics[width=2.2in]{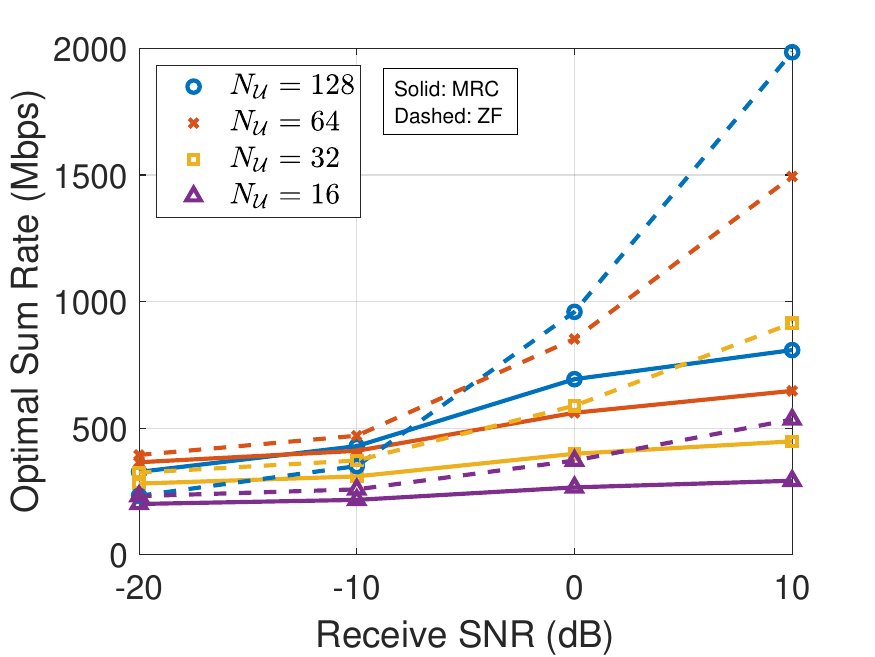}
        \label{fig:opt_sum_rate_vmax5}
    }\hspace{-0.8cm}
    \subfigure[$V_{\textrm{max}} = 25\,\textrm{m/s}$]
    {
        \includegraphics[width=2.2in]{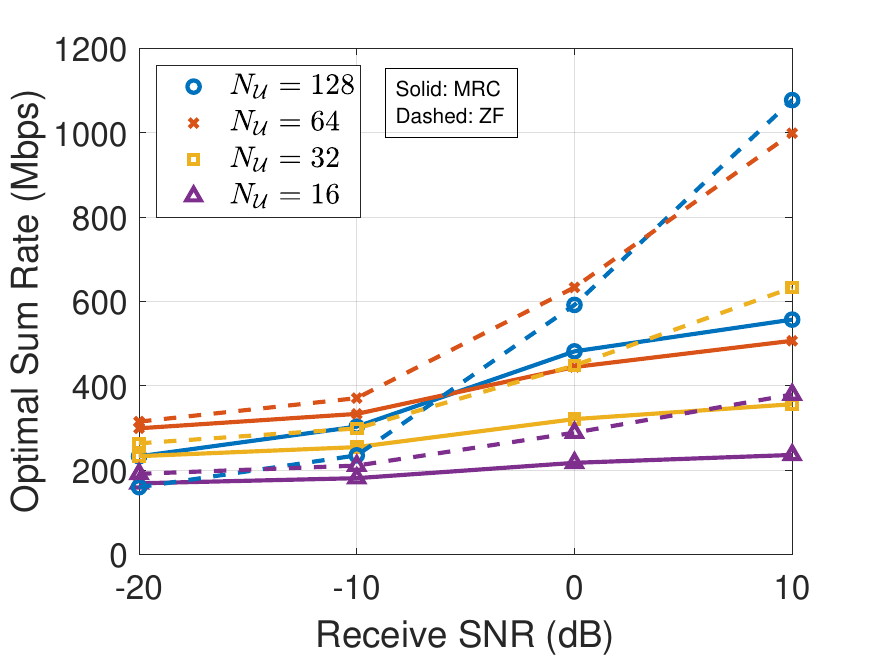}
        \label{fig:opt_sum_rate_vmax25}
    }\hspace{-0.8cm}
    \subfigure[$V_{\textrm{max}} = 100\,\textrm{m/s}$]
    {
        \includegraphics[width=2.25in]{OPT_SUM_RATE_VMAX5.pdf}
        \label{fig:opt_sum_rate_vmax100}
    }
    \caption{Plots of optimal system sum-rate \emph{vs.} receive SNR for ZF and MRC receiver across different values of $N_{\mathcal{U}}$. For a fixed value of $N_{\mathcal{U}}$, the performance of MRC and ZF are comparable at sufficiently low SNR. However, at higher SNR values or when the maximum rate obtainable across different values of $N_{\mathcal{U}}$ is considered, ZF outperforms MRC.}
    \label{fig:opt_sum_rate}
\end{figure}

\section{Conclusions}\label{sec:conclusion}

In this paper, we studied the uplink performance of a multi-user massive MIMO-OFDM cellular system when mobile users transmit data on multiple contiguous subcarriers. We derived an expression for the inter carrier interference power and showed that when the total number of subcarriers in the system is large, each subcarrier incurs a fixed amount of ICI power regardless of the UEs' power distribution across the subcarriers. 
We used the idea of frequency-domain channel coherence to present a pilot allocation scheme to reduce the training overhead involved in channel estimation. Then, we considered zero-forcing combining and maximal-ratio combining at the AP and derived  expressions for the uplink achievable sum-rate. The simple, closed-form expressions provided interesting insights into the core trade-offs involved in MU-mMIMO-OFDM systems in the presence of channel aging. For both receivers, we showed that having a subcarrier serve more number of UEs generally results in higher system sum-rate values, especially for low user mobility. However, such gains in sum-rate performance saturate at a point, beyond which assigning more users to a subcarrier does not provide any further improvement.

In this work, when using the notion of coherence bandwidth, we have assumed that the channel is constant across a set of subcarriers. In practice, the channel coherence bandwidth is usually defined as a set of contiguous subcarriers over which the cross-correlation in the channels across subcarriers is sufficiently high, say $0.7$. Within such a set of subcarriers, the channel will not be constant, but can be closely approximated using a deterministic function, e.g., via linear interpolation across subcarriers. In order to perform linear interpolation, however, we need channel estimates on at least two subcarriers within the coherence bandwidth. With a little bookkeeping effort, the framework in our work can be extended to this scenario also - instead of having each user transmit pilot symbols on one subcarrier within the coherence bandwidth, we can allot \emph{two} subcarriers to each user for pilot transmission. Future work can also consider extending the analysis in this paper to multi-cell systems, accounting for inter-cell interference and pilot contamination. One can also develop channel prediction methods that alleviate the effect of channel aging in OFDM systems while accounting for ICI.

\appendices

\section{Proof of Theorem 1} \label{app:proof_theorem_1}

Equation \eqref{eq:rxzf3} can be rewritten in terms of the normalized channel estimate  $\hat{\mathbf{Z}}_{i} = \sqrt{\frac{1}{\sigma _{\hat{\textrm {{h}}}}^{2}}}\hat{\mathbf{H}}_{i}$ as
\begin{align}
    \mathbf {y}_{i}\left [{n}\right] =& \, \sqrt {P_{\textrm {T}}\,\sigma _{\hat{\textrm {{h}}}}^{4}}\left(\hat{\mathbf{Z}}_{i}\right)^{\dagger} \hat{\mathbf{Z}}_{i} \boldsymbol {\Lambda }_{i}\left [{n}\right ]\mathbf{D}_{\eta_{i}}^{1/2}\mathbf {x}_{i}\left [{n}\right] - 
    \sqrt {P_{\textrm {T}}\,\sigma _{\hat{\textrm {{h}}}}^{2}} \left (\hat{\mathbf{Z}}_{i}\right)^{\dagger} \mathbf {G}_{i}^{\mathcal {P}} \boldsymbol {\Lambda }_{i}\left [{n}\right ]\mathbf{D}_{\eta_{i}}^{1/2}\mathbf {x}_{i}\left [{n}\right ] \nonumber \\ \quad& +
    \, \sqrt {P_{\textrm {T}}\,\sigma _{\hat{\textrm {{h}}}}^{2}} \left (\hat{\mathbf{Z}}_{i}\right )^{\dagger} \mathbf {G}_{i}^{\mathcal {D}}\left [{n}\right ]\mathbf{D}_{\eta_{i}}^{1/2}\mathbf {x}_{i}\left [{n}\right ] +
    \sqrt{\sigma _{\hat{\textrm {{h}}}}^{2}}\,\left (\hat{\mathbf{Z}}_{i}\right )^{\dagger} \left(\mathbf {u}_{i}\left [{n}\right ] + \mathbf {n}_{i}\left [{n}\right ]\right) \label{eq:rxzf4} \\
    =& \, \sqrt {P_{\textrm {T}}\,\sigma _{\hat{\textrm {{h}}}}^{4}} \boldsymbol {\Lambda }_{i}\left [{n}\right ]\mathbf{D}_{\eta_{i}}^{1/2}\mathbf {x}_{i}\left [{n}\right] - 
    \sqrt {P_{\textrm {T}}\,\sigma _{\hat{\textrm {{h}}}}^{2}} \left (\hat{\mathbf{Z}}_{i}\right)^{\dagger} \mathbf {G}_{i}^{\mathcal {P}} \boldsymbol {\Lambda }_{i}\left [{n}\right ]\mathbf{D}_{\eta_{i}}^{1/2}\mathbf {x}_{i}\left [{n}\right ] \nonumber \\ \quad& +
    \, \sqrt {P_{\textrm {T}}\,\sigma _{\hat{\textrm {{h}}}}^{2}} \left (\hat{\mathbf{Z}}_{i}\right )^{\dagger} \mathbf {G}_{i}^{\mathcal {D}}\left [{n}\right ]\mathbf{D}_{\eta_{i}}^{1/2}\mathbf {x}_{i}\left [{n}\right ] +
    \sqrt{\sigma _{\hat{\textrm {{h}}}}^{2}}\,\left (\hat{\mathbf{Z}}_{i}\right )^{\dagger} \left(\mathbf {u}_{i}\left [{n}\right ] + \mathbf {n}_{i}\left [{n}\right ]\right), \label{eq:rxzf5}
\end{align}
where $\sigma_{\hat{\textrm{h}}}^{2}$ denotes the variance of the entries in $\hat{\mathbf{H}}_{i}$. The interference and noise terms, conditioned on the normalized channel estimate $\hat{\mathbf{Z}}_{i}$, are uncorrelated with the desired signal term. Using the approximation in \cite[Lemma 2]{zhang_tcom_2017}, the covariance matrices of the signal terms can be derived. The covariance of the interference due to channel estimation error is given by
\begin{multline}\label{eq:covzf1}
    \textrm{Cov}\Bigg\{\left. \sqrt{P_{\textrm {T}}\,\sigma _{\hat{\textrm {{h}}}}^{2}} \left (\hat{\mathbf{Z}}_{i}\right)^{\dagger} \mathbf {G}_{i}^{\mathcal {P}}  \boldsymbol {\Lambda }_{i}\left [{n}\right ]\mathbf{D}_{\eta_{i}}^{1/2}\mathbf {x}_{i}\left [{n}\right] \right\vert  \hat{\mathbf{Z}}_{i}\Bigg\} \\
    = P_{\textrm{T}}\,\sigma _{\hat{\textrm {{h}}}}^{2}\,\bar {\lambda }\left [{n}\right ] \sum_{k=1}^{N_{\mathcal{U}}} \eta_{ik} \left(\frac{N_{\mathcal{V}}  N_{\mathcal{U}}}{N_{\mathcal{P}} N_{\mathcal{C}}} \sigma_{\textrm{u}}^{2} + \frac{N_{\mathcal{V}} }{N_{\mathcal{P}}P_{\textrm{T}}} \sigma_{\textrm{n}}^{2}\right) \left({\left ({\hat{\mathbf{Z}}_{i}}\right )^{H}\hat{\mathbf{Z}}_{i}}\right )^{-1}.
\end{multline}
The covariance of the interference due to channel aging in \eqref{eq:rxzf5} is given by
\begin{equation}\label{eq:covzf2}
    \textrm{Cov}\Bigg\{\left. \sqrt{P_{\textrm {T}}\,\sigma _{\hat{\textrm {{h}}}}^{2}} \left (\hat{\mathbf{Z}}_{i}\right)^{\dagger} \mathbf {G}_{i}^{\mathcal {D}}\left[n\right]  \mathbf{D}_{\eta_{i}}^{1/2}\mathbf {x}_{i}\left [{n}\right] \right\vert  \hat{\mathbf{Z}}_{i}\Bigg\} =  P_{\textrm{T}}\,\sigma _{\hat{\textrm {{h}}}}^{2}\, \sum_{k=1}^{N_{\mathcal{U}}} \eta_{ik} \, \sigma_{\textrm{h}}^{2} \left(1 - \bar {\lambda }\left [{n}\right ] \right) \left ({\left ({\hat{\mathbf{Z}}_{i}}\right )^{H}\hat{\mathbf{Z}}_{i}}\right )^{-1}.
\end{equation} 
The covariance of the ICI and AWGN term in \eqref{eq:rxzf5} is given by
\begin{equation}\label{eq:covzf3}
    \textrm{Cov}\Bigg\{\left. \sqrt{\sigma _{\hat{\textrm {{h}}}}^{2}}\,\left (\hat{\mathbf{Z}}_{i}\right )^{\dagger} \left(\mathbf {u}_{i}\left [{n}\right ] + \mathbf {n}_{i}\left [{n}\right ]\right) \right\vert  \hat{\mathbf{Z}}_{i}\Bigg\} = \sigma _{\hat{\textrm {{h}}}}^{2}\,\left(\frac{N_{\mathcal{U}} P_{\textrm {T}}}{N_{\mathcal{C}}} \sigma_{\textrm{u}}^{2} + \sigma_{\textrm{n}}^{2}\right) \left ({\left ({\hat{\mathbf{Z}}_{i}}\right )^{H}\hat{\mathbf{Z}}_{i}}\right )^{-1}.
\end{equation}
Finally, the covariance of the desired signal term in \eqref{eq:rxzf5} is given by
\begin{equation}\label{eq:covzf4}
    \textrm{Cov}\Bigg\{\left. \sqrt {P_{\textrm {T}}\,\sigma _{\hat{\textrm {{h}}}}^{4}} \boldsymbol {\Lambda }_{i}\left [{n}\right ]\mathbf{D}_{\eta_{i}}^{1/2}\mathbf {x}_{i}\left [{n}\right]\right\vert \hat{\mathbf{Z}}_{i}\Bigg\} = P_{\textrm {T}}\,\sigma _{\hat{\textrm {{h}}}}^{4}\,\mathbf{D}_{\eta_{i}} \bar {\lambda }\left [{n}\right ] \mathbf {I}_{N_{\mathcal {U}}}.
\end{equation}
We note that the covariances in \eqref{eq:covzf2}, \eqref{eq:covzf3} and \eqref{eq:covzf4} correspond to noise and interference terms that are mutually uncorrelated with each other.

Now, the uplink SINR of UE$_{ik}$'s $n$-th transmission can be found by extracting the $(k,k)$-th element of the covariance matrices in \eqref{eq:covzf1}, \eqref{eq:covzf2}, \eqref{eq:covzf3} and \eqref{eq:covzf4}, respectively, and dividing the variance of the desired signal term by the sum of the variances of the interference plus noise terms. Thus, we get
\begin{equation}
    \textrm{SINR}_{ik}^{zf, ul}[n] = \frac{\eta_{ik} \sigma _{\hat{\textrm {h}}}^{2} \bar {\lambda }\left [{n}\right ]}{\left(\frac{N_{\mathcal {U}}\sigma _{\textrm {u}}^{2}}{N_{\mathcal{C}}} + \frac{\sigma _{\textrm {n}}^{2}}{P_{\textrm {T}}}\right)\left(1+\frac{N_{\mathcal{V}}  \bar{\eta}_{i}}{N_{\mathcal{P}}} \bar {\lambda }\left [{n}\right ]\right) + \bar{\eta}_{i} \sigma _{\textrm {h}}^{2}\left ({1-\bar {\lambda }\left [{n}\right ]}\right)\left ({\left ({\hat{\mathbf{Z}}_{i}}\right )^{H}\hat{\mathbf{Z}}_{i}}\right )^{-1}_{kk}}.
\end{equation}
Here, $\bar{\eta}_{i} = \sum_{k=1}^{N_{\mathcal{U}}} \eta_{ik}$,  ${N_{\mathcal{V}}  = \lceil{N_{\mathcal{C}}/N_{\mathcal{H}}}\rceil}$; $\bar {\lambda }\left [{n}\right ]=\frac {1}{V_{\textrm {max}}}\int _{0}^{V_{\textrm {max}}}J_{0}^{2}\left({\frac {2\pi vf_{\textrm{c}}nT_{\textrm {s}}}{\textrm {c}}}\right)dv$ denotes the expectation of the diagonal entries in $\left ({\boldsymbol {\Lambda }_{i}\left [{n}\right ]}\right )^{2}$ and $\sigma_{\hat{\textrm{h}}}^{2} = \sigma_{\textrm{h}}^{2} + \frac{N_{\mathcal{V}}  N_{\mathcal{U}}}{N_{\mathcal{P}} N_{\mathcal{C}}} \sigma_{\textrm{u}}^{2} + \frac{N_{\mathcal{V}} }{N_{\mathcal{P}}P_{\textrm{T}}} \sigma_{\textrm{n}}^{2}$ denotes the variance of the entries in $\hat{\mathbf{H}}_{i}$.

Now, the achievable uplink rate of UE$_{ik}$'s $n$-th transmission on the $i$-th subcarrier can be computed as $\textrm{C}^{\textrm{zf, ul}}_{ik}[n] \geq \mathbb{E}\{\Delta f \log_{2}\left(1 + \textrm{SINR}_{ik}^{\textrm{zf, ul}}[n]\right)\}$\cite{marzetta_book_2016}, where the expectation is with respect to the term $\frac{1}{\left(\left ({\left ({\hat{\mathbf{Z}}}_{i}\right )^{H}\hat{\mathbf{Z}}}_{i}\right )^{-1}\right)_{kk}}$. The expectation operator can be taken inside the logarithm by virtue of Jensen's inequality. Then, using the observation that  $\mathbb{E}\Bigg\{\frac{1}{\left(\left ({\left ({\hat{\mathbf{Z}}}_{i}\right )^{H}\hat{\mathbf{Z}}}_{i}\right )^{-1}\right)_{kk}}\Bigg\} = N_{\mathcal{B}} - N_{\mathcal{U}} + 1,  k = 1, \dots, N_{\mathcal{U}}$ from \cite{marzetta_book_2016}, we obtain~\eqref{eq:ratezf}.

\section{Proof of Theorem 2} \label{app:proof_theorem_2}
Equation \eqref{eq:yxmrc1} can be rewritten in terms of the normalized channel estimate $\hat{\mathbf{Z}}_{i}$ as 
\begin{equation}\label{eq:yxmrc2}
\begin{aligned}
    \mathbf {y}_{i}\left [{n}\right] = \sqrt{P_{\textrm {T}}\,\sigma _{\hat{\textrm {{h}}}}^{4}}\,\left(\hat{\mathbf{Z}}_{i}\right)^{H} \hat{\mathbf{Z}}_{i} \boldsymbol {\Lambda }_{i}\left [{n}\right ]\mathbf{D}_{\eta_{i}}^{1/2}\mathbf {x}_{i}\left [{n}\right ] - 
    \sqrt {P_{\textrm {T}}\,\sigma _{\hat{\textrm {{h}}}}^{2}}\,\left (\hat{\mathbf{Z}}_{i}\right)^{H} \mathbf {G}_{i}^{\mathcal {P}} \boldsymbol {\Lambda }_{i}\left [{n}\right ]\mathbf{D}_{\eta_{i}}^{1/2}\mathbf {x}_{i}\left [{n}\right ] \\ +
    \, \sqrt {P_{\textrm {T}}\,\sigma _{\hat{\textrm {{h}}}}^{2}}\,\left (\hat{\mathbf{Z}}_{i}\right )^{H} \mathbf {G}_{i}^{\mathcal {D}}\left [{n}\right ]\mathbf{D}_{\eta_{i}}^{1/2}\mathbf {x}_{i}\left [{n}\right ] +
    \sqrt{\sigma_{\hat{\textrm {{h}}}}^{2}} \, \left (\hat{\mathbf{Z}}_{i}\right)^{H} \left(\mathbf {u}_{i}\left [{n}\right ] + \mathbf {n}_{i}\left [{n}\right ]\right).
\end{aligned}
\end{equation}
Now, the $k$-th element in $\mathbf {y}_{i}\left [{n}\right]$ is 
\begin{equation}\label{eq:yxmrc3}
\begin{aligned}
    \Big\{\mathbf{y}_{i}[n]\Big\}_{k} = \sqrt{P_{\textrm {T}}\,\sigma _{\hat{\textrm {{h}}}}^{4}}\,\hat{\mathbf{z}}_{ik}^{H} \,\hat{\mathbf{Z}}_{i} \boldsymbol {\Lambda }_{i}\left [{n}\right ]\mathbf{D}_{\eta_{i}}^{1/2}\mathbf {x}_{i}\left [{n}\right ] - 
    \sqrt {P_{\textrm {T}}\,\sigma _{\hat{\textrm {{h}}}}^{2}}\, \hat{\mathbf{z}}_{ik}^{H} \, \mathbf {G}_{i}^{\mathcal {P}} \boldsymbol {\Lambda }_{i}\left [{n}\right ]\mathbf{D}_{\eta_{i}}^{1/2}\mathbf {x}_{i}\left [{n}\right ] \\ +
    \, \sqrt {P_{\textrm {T}} \, \sigma_{\hat{\textrm {{h}}}}^{2}} \, \hat{\mathbf{z}}_{ik}^{H} \mathbf {G}_{i}^{\mathcal {D}}\left [{n}\right ]\mathbf{D}_{\eta_{i}}^{1/2}\mathbf {x}_{i}\left [{n}\right ] +
    \sqrt{\sigma_{\hat{\textrm {{h}}}}^{2}} \, \hat{\mathbf{z}}_{ik}^{H} \left(\mathbf {u}_{i}\left [{n}\right ] + \mathbf {n}_{i}\left [{n}\right ]\right),
\end{aligned}
\end{equation}
where $\hat{\mathbf{z}}_{ik}$ denotes the $k$-th column of $\hat{\mathbf{Z}}_{i}$. To obtain an expression for the achievable uplink SINR, we employ the \textit{use and then forget CSI} approach from \cite{marzetta_book_2016} in which a first party performs MR combining using knowledge of the channel estimate and passes the signal to another party that processes the signal based on the expected value of the equivalent channel. We invoke this approach by rewriting \eqref{eq:yxmrc3} as 
\begin{align}\label{eq:unfcsi}
    \frac{1}{\sqrt{N_{\mathcal{B}}}} \Big\{\mathbf{y}_{i}[n]\Big\}_{k} =& \, \sqrt{\frac{P_{\textrm{T}}\,\sigma _{\hat{\textrm {{h}}}}^{4} \rho_{ik}^{2}[n] \eta_{ik}}{N_{\mathcal{B}}}} \, \mathbb{E}\Big\{||\hat{\mathbf{z}}_{ik}||^{2}\Big\} \, x_{ik}[n] \nonumber \\
    +& \,\frac{1}{\sqrt{N_{\mathcal{B}}}} \left(\sqrt{\sigma_{\hat{\textrm {{h}}}}^{2}} \,\hat{\mathbf{z}}_{ik}^{H} \, \mathbf{u}_{i}[n] + \sqrt{\sigma_{\hat{\textrm {{h}}}}^{2}} \, \hat{\mathbf{z}}_{ik}^{H} \, \mathbf{n}_{i}[n]\right) \nonumber \\
    +& \,\frac{1}{\sqrt{N_{\mathcal{B}}}} \left(\sqrt{P_{\textrm {T}}\,\sigma_{\hat{\textrm {{h}}}}^{2}} \, \hat{\mathbf{z}}_{ik}^{H} \,\mathbf {G}_{i}^{\mathcal {D}}\left [{n}\right ]\mathbf{D}_{\eta_{i}}^{1/2}\mathbf {x}_{i}\left [{n}\right]- \sqrt{P_{\textrm {T}}\,\sigma_{\hat{\textrm {{h}}}}^{2}} \, \hat{\mathbf{z}}_{ik}^{H}\,\mathbf {G}_{i}^{\mathcal{P}}\boldsymbol {\Lambda }_{i}\left [{n}\right ]\mathbf{D}_{\eta_{i}}^{1/2}\mathbf {x}_{i}\left [{n}\right]\right) \nonumber \\
    +& \, \frac{1}{\sqrt{N_{\mathcal{B}}}} \left(\hat{\mathbf{z}}_{ik}^{H} \, \sum_{\substack{k'=1 \\ k' \neq k}}^{N_{\mathcal{U}}}  \sqrt{P_{\textrm {T}}\,\sigma _{\hat{\textrm {{h}}}}^{4}\,\rho_{ik}^{2}[n]\, \eta_{ik'}} \, \hat{\mathbf{z}}_{ik'} x_{ik'}\left [{n}\right ]\right) \nonumber \\
    +& \, \sqrt{\frac{P_{\textrm{T}}\,\sigma _{\hat{\textrm {{h}}}}^{4} \rho_{ik}^{2}[n] \eta_{ik}}{N_{\mathcal{B}}}} \left(||\hat{\mathbf{z}}_{ik}||^{2} - \mathbb{E}\Big\{||\hat{\mathbf{z}}_{ik}||^{2}\Big\}\right)x_{ik}[n].
\end{align}
The mean-square value of the first term in the above expression, which represents the desired signal, is
\begin{equation}\label{eq:covmrc1}
        \frac{P_{\textrm{T}}\,\sigma _{\hat{\textrm {{h}}}}^{4} \bar {\lambda }\left [{n}\right] \eta_{ik}}{N_{\mathcal{B}}} \left(\mathbb{E}\Big\{||\hat{\mathbf{z}}_{ik}||^{2}\Big\}\right)^{2} = N_{\mathcal{B}}P_{\textrm{T}}\,\sigma _{\hat{\textrm {{h}}}}^{4} \bar {\lambda }\left [{n}\right] \eta_{ik}.
\end{equation}
The second and third terms in \eqref{eq:unfcsi} contain interference from ICI, AWGN, channel aging and channel estimation error, and have variance given by
\begin{align}\label{eq:covmrc2}
        \sigma_{\hat{\textrm{{h}}}}^{2}\sigma_{\textrm{u}}^{2}\frac{N_{\mathcal{U}}P_{\textrm{T}}}{N_{\mathcal{C}}} + \sigma_{\hat{\textrm{{h}}}}^{2} \sigma_{\textrm{n}}^{2} + P_{\textrm{T}}\sigma_{\hat{\textrm{{h}}}}^{2}\sigma_{\textrm{h}}^{2}\left(1 - \bar {\lambda }\left [{n}\right]\right)\sum_{k'=1}^{N_{\mathcal{U}}}\eta_{ik'} + \, P_{\textrm{T}}\sigma_{\hat{\textrm{{h}}}}^{2} \bar {\lambda }\left [{n}\right]\left(\frac{N_{\mathcal{V}}  N_{\mathcal{U}}}{N_{\mathcal{P}} N_{\mathcal{C}}} \sigma_{\textrm{u}}^{2} + \frac{N_{\mathcal{V}} }{N_{\mathcal{P}}P_{\textrm{T}}} \sigma_{\textrm{n}}^{2}\right)\sum_{k'=1}^{N_{\mathcal{U}}}\eta_{ik'}.
\end{align}    
The fourth term in \eqref{eq:unfcsi} represents channel non-orthogonality and has variance given by
\begin{align}
        \frac{1}{N_{\mathcal{B}}}\textrm{Var}\Bigg\{\hat{\mathbf{z}}_{ik}^{H} \, \sum_{\substack{k'=1 \\ k' \neq k}}^{N_{\mathcal{U}}}  \sqrt{P_{\textrm {T}}\,\sigma _{\hat{\textrm {{h}}}}^{4}\,\bar {\lambda }\left [{n}\right]\, \eta_{ik'}} \, \hat{\mathbf{z}}_{ik'} \mathbf {x}_{ik'}\left [{n}\right ]\Bigg\} =& \, \frac{P_{\textrm {T}}\,\sigma _{\hat{\textrm {{h}}}}^{4}}{N_{\mathcal{B}}} \bar {\lambda }\left [{n}\right] \sum_{\substack{k'=1 \\ k' \neq k}}^{N_{\mathcal{U}}} \, \eta_{ik'} \mathbb{E}\Big\{|\hat{\mathbf{z}}_{ik}^{H}\hat{\mathbf{z}}_{ik'}|\Big\} \\
        =& \, P_{\textrm {T}}\,\sigma _{\hat{\textrm {{h}}}}^{4} \bar {\lambda }\left [{n}\right] \sum_{\substack{k'=1 \\ k' \neq k}}^{N_{\mathcal{U}}} \, \eta_{ik'}\label{eq:covmrc3}.
\end{align}
The fifth term in \eqref{eq:unfcsi} represents the beamforming gain uncertainty. Its variance is given by 
\begin{align}
        \frac{P_{\textrm {T}}\,\sigma_{\hat{\textrm {{h}}}}^{4}\,\bar {\lambda }\left [{n}\right]\, \eta_{ik}}{N_{\mathcal{B}}} &\textrm{Var} \Big\{\left(||\hat{\mathbf{z}}_{ik}||^{2} - \mathbb{E}\Big\{||\hat{\mathbf{z}}_{ik}||^{2}\right)x_{ik}\Big\} \nonumber \\ 
        =& \, \frac{P_{\textrm {T}}\,\sigma _{\hat{\textrm {{h}}}}^{4}\,\bar {\lambda }\left [{n}\right]\, \eta_{ik}}{N_{\mathcal{B}}} \left(\mathbb{E}\Big\{||\hat{\mathbf{z}}_{ik}||^{4}\Big\} - \left(\mathbb{E}\Big\{||\hat{\mathbf{z}}_{ik}||^{2}\Big\}\right)^{2}\right) = P_{\textrm {T}}\,\sigma _{\hat{\textrm {{h}}}}^{4}\,\bar {\lambda }\left [{n}\right]\, \eta_{ik}\label{eq:covmrc4}.
\end{align}
The uplink SINR of UE$_{ik}$'s $n$-th transmission on the $i$-th subcarrier is obtained by dividing the variance in \eqref{eq:covmrc1} by the sum of the variances in \eqref{eq:covmrc2}, \eqref{eq:covmrc3} and \eqref{eq:covmrc4}. Hence, we obtain
\begin{equation}
    \textrm{SINR}_{ik}^{\textrm{mrc, ul}}\left[n\right] = \frac{N_{\mathcal{B}}\,\eta_{ik}\, \sigma _{\hat{\textrm {h}}}^{2}\,\bar {\lambda }\left [{n}\right ]}{\left(\frac{N_{\mathcal {U}}\sigma _{\textrm {u}}^{2}}{N_{\mathcal{C}}} + \frac{\sigma _{\textrm {n}}^{2}}{P_{\textrm {T}}}\right)\left(1+\frac{N_{\mathcal{V}}  \bar{\eta}_{i}}{N_{\mathcal{P}}} \bar {\lambda }\left [{n}\right ]\right) + \bar{\eta}_{i} \sigma _{\textrm {h}}^{2}\left ({1-\bar {\lambda }\left [{n}\right ]}\right ) + \sigma _{\hat{\textrm {{h}}}}^{2}\, \bar{\eta}_{i}\, \bar {\lambda }\left [{n}\right]},
\end{equation}
where $\bar{\eta}_{i} = \sum_{k=1}^{N_{\mathcal{U}}} \eta_{ik}$;  ${N_{\mathcal{V}}  = \lceil{N_{\mathcal{C}}/N_{\mathcal{H}}}\rceil}$, $\bar {\lambda }\left [{n}\right ]=\frac {1}{V_{\textrm {max}}}\int _{0}^{V_{\textrm {max}}}J_{0}^{2}\left({\frac {2\pi vf_{\textrm{c}}nT_{\textrm {s}}}{\textrm {c}}}\right)dv$ denotes the expectation of the diagonal entries in $\left ({\boldsymbol {\Lambda }_{i}\left [{n}\right ]}\right )^{2}$ and $\sigma_{\hat{\textrm{h}}}^{2} = \sigma_{\textrm{h}}^{2} + \frac{N_{\mathcal{V}}  N_{\mathcal{U}}}{N_{\mathcal{P}} N_{\mathcal{C}}} \sigma_{\textrm{u}}^{2} + \frac{N_{\mathcal{V}} }{N_{\mathcal{P}}P_{\textrm{T}}} \sigma_{\textrm{n}}^{2}$ denotes the variance of the entries in $\hat{\mathbf{H}}_{i}$. The achievable uplink rate is then computed as $\textrm{C}_{ik}^{\textrm{mrc, ul}} \left[n\right] \geq \log_{2}\left(1 + \textrm{SINR}_{ik}^{\textrm{mrc, ul}}\left[n\right]\right)$\cite{marzetta_book_2016}. This leads us to \eqref{eq:ratemrc}.

\ifCLASSOPTIONcaptionsoff
  \newpage
\fi

\bibliographystyle{IEEEtran}
\bibliography{IEEEabrv.bib, bibliography.bib}

\begin{thebibliography}{10}
\providecommand{\url}[1]{#1}
\csname url@samestyle\endcsname
\providecommand{\newblock}{\relax}
\providecommand{\bibinfo}[2]{#2}
\providecommand{\BIBentrySTDinterwordspacing}{\spaceskip=0pt\relax}
\providecommand{\BIBentryALTinterwordstretchfactor}{4}
\providecommand{\BIBentryALTinterwordspacing}{\spaceskip=\fontdimen2\font plus
\BIBentryALTinterwordstretchfactor\fontdimen3\font minus
  \fontdimen4\font\relax}
\providecommand{\BIBforeignlanguage}[2]{{%
\expandafter\ifx\csname l@#1\endcsname\relax
\typeout{** WARNING: IEEEtran.bst: No hyphenation pattern has been}%
\typeout{** loaded for the language `#1'. Using the pattern for}%
\typeout{** the default language instead.}%
\else
\language=\csname l@#1\endcsname
\fi
#2}}
\providecommand{\BIBdecl}{\relax}
\BIBdecl

\bibitem{bjrnson_magazine_2016}
E.~{Björnson}, E.~G. {Larsson}, and T.~L. {Marzetta}, ``Massive {MIMO}: ten
  myths and one critical question,'' \emph{{IEEE} Commun. Mag.}, vol.~54,
  no.~2, pp. 114--123, 2016.

\bibitem{marzetta_twc_2010}
T.~L. {Marzetta}, ``Noncooperative cellular wireless with unlimited numbers of
  base station antennas,'' \emph{{IEEE} Trans. Wireless Commun.}, vol.~9,
  no.~11, pp. 3590--3600, 2010.

\bibitem{hoydis_jsac_2013}
J.~{Hoydis}, S.~{ten Brink}, and M.~{Debbah}, ``Massive {MIMO} in the {UL/DL}
  of cellular networks: How many antennas do we need?'' \emph{{IEEE} J. Select.
  Areas Commun.}, vol.~31, no.~2, pp. 160--171, 2013.

\bibitem{marzetta_book_2016}
T.~{Marzetta}, E.~Larsson, H.~Yang, and H.~Ngo, \emph{Fundamentals of Massive
  MIMO}.\hskip 1em plus 0.5em minus 0.4em\relax Cambridge University Press,
  2016.

\bibitem{xiao_tvt_2015}
X.~{Xiao}, X.~{Tao}, and J.~{Lu}, ``Energy-efficient resource allocation in
  {LTE}-based {MIMO-OFDMA} systems with user rate constraints,'' \emph{{IEEE}
  Trans. Veh. Technol.}, vol.~64, no.~1, pp. 185--197, 2015.

\bibitem{ng_tcom_2012}
D.~W.~K. {Ng}, E.~S. {Lo}, and R.~{Schober}, ``Dynamic resource allocation in
  {MIMO-OFDMA} systems with full-duplex and hybrid relaying,'' \emph{{IEEE}
  Trans. Commun.}, vol.~60, no.~5, pp. 1291--1304, 2012.

\bibitem{xu_tcom_2013}
Z.~{Xu}, C.~{Yang}, G.~Y. {Li}, S.~{Zhang}, Y.~{Chen}, and S.~{Xu},
  ``Energy-efficient configuration of spatial and frequency resources in
  {MIMO-OFDMA} systems,'' \emph{{IEEE} Trans. Commun.}, vol.~61, no.~2, pp.
  564--575, 2013.

\bibitem{akp_tvt_2017}
A.~K. {Papazafeiropoulos}, ``Impact of general channel aging conditions on the
  downlink performance of massive {MIMO},'' \emph{{IEEE} Trans. Veh. Technol.},
  vol.~66, no.~2, pp. 1428--1442, 2017.

\bibitem{akp_twc_2015}
A.~K. {Papazafeiropoulos} and T.~{Ratnarajah}, ``Deterministic equivalent
  performance analysis of time-varying massive {MIMO} systems,'' \emph{{IEEE}
  Trans. Commun.}, vol.~14, no.~10, pp. 5795--5809, 2015.

\bibitem{chopra_twc_2018}
R.~{Chopra}, C.~R. {Murthy}, H.~A. {Suraweera}, and E.~G. {Larsson},
  ``Performance analysis of {FDD} massive {MIMO} systems under channel aging,''
  \emph{{IEEE} Trans. Wireless Commun.}, vol.~17, no.~2, pp. 1094--1108, 2018.

\bibitem{truong_jcn_2013}
K.~T. {Truong} and R.~W. {Heath}, ``Effects of channel aging in massive {MIMO}
  systems,'' \emph{J. Commun. Netw.}, vol.~15, no.~4, pp. 338--351, 2013.

\bibitem{chopra_letters_2016}
R.~{Chopra}, C.~R. {Murthy}, and H.~A. {Suraweera}, ``On the throughput of
  large {MIMO} beamforming systems with channel aging,'' \emph{{IEEE} Signal
  Processing Lett.}, vol.~23, no.~11, pp. 1523--1527, 2016.

\bibitem{kong_tcom_2015}
C.~Kong, C.~Zhong, A.~K. Papazafeiropoulos, M.~Matthaiou, and Z.~Zhang,
  ``Sum-rate and power scaling of massive {MIMO} systems with channel aging,''
  \emph{{IEEE} Trans. Commun.}, vol.~63, no.~12, pp. 4879--4893, 2015.

\bibitem{kim_tcom_2021}
H.~Kim, S.~Kim, H.~Lee, C.~Jang, Y.~Choi, and J.~Choi, ``Massive {MIMO} channel
  prediction: Kalman filtering vs. machine learning,'' \emph{{IEEE} Trans.
  Commun.}, vol.~69, no.~1, pp. 518--528, 2021.

\bibitem{Chopra_TSP_2021}
R.~Chopra and C.~R. Murthy, ``Data aided {MSE}-optimal time varying channel
  tracking in massive {MIMO} systems,'' \emph{{IEEE} Trans. Signal Processing},
  {Accepted, Jun. 2021}.

\bibitem{Chopra_ComLet_2021}
R.~Chopra, C.~R. Murthy, and A.~K. Papazafeiropoulos, ``Uplink performance
  analysis of cell-free m{MIMO} systems under channel aging,'' \emph{{IEEE}
  Commun. Lett.}, 2021.

\bibitem{zhang_tcom_2017}
Z.~{Zhang}, C.~{Jiao}, and C.~{Zhong}, ``Impact of mobility on the uplink sum
  rate of {MIMO-OFDMA} cellular systems,'' \emph{{IEEE} Trans. Commun.},
  vol.~65, no.~10, pp. 4218--4231, 2017.

\bibitem{dai_icassp_2006}
H.~{Dai}, ``Distributed versus co-located {MIMO} systems with correlated fading
  and shadowing,'' in \emph{2006 IEEE International Conference on Acoustics
  Speech and Signal Processing Proceedings}, vol.~4, 2006, pp. IV--IV.

\bibitem{jakes_wiley_1974}
W.~C. Jakes and D.~C. Cox, Eds., \emph{Microwave Mobile Communications},
  2nd~ed.\hskip 1em plus 0.5em minus 0.4em\relax IEEE Press, New York: IEEE
  Press, 1994.

\bibitem{vu_jsac_2007}
M.~{Vu} and A.~{Paulraj}, ``On the capacity of {MIMO} wireless channels with
  dynamic {CSIT},'' \emph{{IEEE} J. Select. Areas Commun.}, vol.~25, no.~7, pp.
  1269--1283, 2007.

\bibitem{zheng_twc_2021}
J.~Zheng, J.~Zhang, E.~Björnson, and B.~Ai, ``Impact of channel aging on
  cell-free massive {MIMO} over spatially correlated channels,'' \emph{{IEEE}
  Trans. Wireless Commun.}, vol.~20, no.~10, pp. 6451--6466, 2021.

\bibitem{bjrnson_book_2017}
E.~{Bj{\"o}rnson}, J.~{Hoydis}, L.~{Sanguinetti} \emph{et~al.}, ``Massive
  {MIMO} networks: Spectral, energy, and hardware efficiency,''
  \emph{Foundations and Trends{\textregistered} in Signal Processing}, vol.~11,
  no. 3-4, pp. 154--655, 2017.

\end{thebibliography}

\end{document}